\documentclass[12pt]{article}

\usepackage{fullpage}
\usepackage{amsmath,amsfonts,amssymb}
\usepackage{amsthm}
\usepackage{pdfpages}
\usepackage{float}

\newtheorem{theorem}{Theorem}[section]
\newtheorem{corollary}[theorem]{Corollary}
\newtheorem{lemma}[theorem]{Lemma}
\newtheorem{claim}[theorem]{Claim}
\newtheorem{fact}[theorem]{Fact}

\newtheorem{observation}[theorem]{Observation}

\newenvironment{proofof}[1]{ {\sc Proof of #1.}\/}{\qedsymbol}

\renewcommand{\qedsymbol}{\ensuremath{\blacksquare}}

\newcommand{\comment}[1]{}
\newcommand{\remove}[1]{}
\newcommand{\suppress}[1]{}

\newcommand{\opt}{{\hbox{\sc opt}}}
\newcommand{\alg}{{\hbox{\sc alg}}}
\newcommand{\LP}{{\hbox{\sc lp}}}

\newcommand{\DP}{{\hbox{\sc dp}}}

\begin{document}

\title{An Optimal Randomized Online Algorithm for Reordering Buffer Management}

\author{
Noa Avigdor-Elgrabli\thanks{Computer Science Department,
Technion---Israel Institute of Technology, Haifa 32000, Israel.
Email: {\tt noaelg@cs.technion.ac.il}}
\and
Yuval Rabani\thanks{The Rachel and Selim Benin School of 
Computer Science and Engineering, The Hebrew University of Jerusalem,
Jerusalem 91904, Israel. Email: {\tt yrabani@cs.huji.ac.il}. Research
supported by Israel Science Foundation grant number 856-11 and
by the Israeli Center of Excellence on Algorithms.}
}

\date{\today}

\setcounter{footnote}{3}
\maketitle

\begin{abstract}
We give an $O(\log\log k)$-competitive randomized online
algorithm for reordering buffer management, where $k$ is the
buffer size. Our bound matches the lower bound of
Adamaszek et al. (STOC 2011). Our algorithm has two stages
which are executed online in parallel. The first stage computes
deterministically a feasible fractional solution to an LP relaxation 
for reordering buffer management. The second stage ``rounds" 
using randomness the fractional solution. The first stage 
is based on the online primal-dual schema, combined with a 
dual fitting argument. As multiplicative weights steps and dual
fitting steps are interleaved and in some sense conflicting, 
combining them is challenging. We also note that we apply
the primal-dual schema to a relaxation with mixed packing
and covering constraints. We pay the $O(\log\log k)$ competitive 
factor for the gap between the computed LP solution and the
optimal LP solution. The second stage
gives an online algorithm that converts the LP solution to
an integral solution, while increasing the cost by an $O(1)$ 
factor. This stage generalizes recent results that gave a
similar approximation factor for rounding the LP solution, 
albeit using an offline rounding algorithm.
\end{abstract}


\section{Introduction}

In the reordering buffer management problem (RBM) an input
sequence of colored items arrives online, and has to be rescheduled 
in a permuted output sequence of the same items, with the help 
of a buffer that can hold $k$ items. The items enter the buffer 
in their order of arrival. When the buffer is full, one color present
in the buffer must be chosen, and the items of this color in the
buffer, followed by any new items of the same color encountered 
along the way, are scheduled in the output sequence one item
per time slot, making room for new input items to enter the
buffer. The choice of color is made before future input items 
are revealed. Choosing a color
and evicting items is repeated until we reach the end of the 
input sequence and we empty the buffer. The objective is to 
minimize the total number of color changes between consecutive 
items in the output schedule. This seemingly simple model, 
introduced in~\cite{RSW02}, formalizes a wide scope of resource 
management problems in production engineering, logistics, 
computer systems, network optimization, and information retrieval 
(see, e.g.,~\cite{RSW02,BB02,KRSW04,GSV04}). Moreover, 
beyond its simplicity, elegance, and applicability, the problem
turns out to be challenging, and it captures some new and 
fundamental issues in online computing. We note that the
offline version of RBM is NP-hard~\cite{AKM10,CMSS10}, and
there is a polynomial time $O(1)$-approximation algorithm~\cite{AR13}.

This paper resolves the randomized competitive ratio of RBM.
We design a randomized online RBM algorithm and prove that
its competitive ratio is $O(\log\log k)$. This matches the
recent lower bound of $\Omega(\log\log k)$ of Adamaszek
et al.~\cite{ACER11}.
All previous online algorithms for RBM are deterministic. A
sequence of papers~\cite{RSW02,EW05,AR10,ACER11}
culminated in an 
$O(\sqrt{\log k})$-competitive algorithm~\cite{ACER11},
nearly matching the deterministic lower bound in the same
paper. Thus, our work is the first to demonstrate an
exponential gap between the deterministic and the randomized 
competitive ratio of RBM.

In essence, our algorithm is an implementation of the primal-dual 
schema, and more specifically of the multiplicate weights update
method (see~\cite{BN09} for a survey of its use in online computing,
and~\cite{AHK12} for a general survey of the method). 
We compute online a feasible solution 
to an LP relaxation
for RBM. This part is done deterministically, and it uses the same
relaxation as in our past paper~\cite{AR10}. As we compute the LP 
solution, we feed it to an online rounding algorithm, which generates 
an integral solution on-the-fly. This part uses randomness, and is
motivated by our recent paper~\cite{AR13} that gives a (deterministic)
constant factor polynomial time approximation algorithm for RBM.

One of the interesting aspects of our result is that we apply the 
multiplicative weights update method to a bipartite perfect matching-style 
linear program. Essentially all prior online results 
using this method (e.g.~\cite{AAABN09,BBN12,BBMN11,ACER12}) 
were derived through relaxations that are packing 
or covering linear programs, or small variations thereof (such as 
having additional box constraints). Another interesting aspect of 
our result is that, for reasons explained below, we cannot apply this 
method in its pure form. We must combine it with a dual fitting argument 
that is similar in spirit to that in~\cite{AR10}. Combining the two
conflicting approaches into a hybrid primal-dual algorithm and
proof is the main technical challenge of our work.

What's unique about RBM and what separates it from other reputed
online problems is the following. Usually, when an online algorithm 
makes an irrevocable decision, this may change the state of the system
in a way that changes the cost of future decisions, but it only sets
the current output. In contrast, when an RBM algorithm makes a
decision to evict a color block from the buffer, this may decide the 
output for many steps ahead, as only one item can be evicted in one
step. What's even worse, as items that are already in the buffer when
the decision is made are evicted, new items of this color that arrive
along the way can be appended to the evicted sequence, so the algorithm
is perhaps deciding now how to handle future input it hasn't yet seen.
Moreover, the size of the color block being evicted could have a dramatic
influence on the length of the sequence being evicted, and hence on
future cost: if two algorithms differ by just one item between their buffer 
contents at some point in time, and they both decide to evict a color in 
which they differ, the algorithm with one less item might evict a much 
shorter sequence than the algorithm with one more item.

Very recently, Adamaszek et al.~\cite{ACER12} proposed a 
problem they call {\em buffer scheduling for block devices},
which is a variant of RBM that differs from it in one crucial
aspect. In block devices, while evicting a color from the buffer, 
new arriving items of that color cannot be appended to the 
output sequence without incurring additional cost. For example,
if the entire input sequence consists of a single color, an RBM 
solution pays $1$, while a block device solution pays 
approximately $\frac{n}{k}$. Thus, their problem eliminates 
the issue of making decisions regarding unknown future items, 
but it still has to cope with the problem of making decisions 
regarding future output steps. 
They give an $O(\log\log k)$-competitive
randomized online algorithm for buffer scheduling for
block devices that implements the multiplicative weights
update method (see~\cite{BN09}), using a covering LP formulation.

Their result, which motivated our work and
influenced part of it, overcomes the issue of deciding on
future output steps by employing a resource augmentation 
argument, adapted from~\cite{ERW09}. In~\cite{EW05,ERW09} 
it is proved (for RBM, but the same proof applies to block
devices) that if we replace a buffer of size $k$ by a buffer of 
size $\frac{k}{4}$, then the optimal value of a solution cannot 
increase by more than an $O(\log k)$ factor. If it were possible 
to reduce this factor to $O(1)$ for some constant factor decrease
in the buffer size, we could derive our results with ease using 
a fairly simple implementation of multiplicative weights. 
Unfortunately, the $O(\log k)$ factor is tight~\cite{Abo08}.
What is used in~\cite{ACER12} is a straightforward generalization 
of the proof in~\cite{ERW09}, which shows (for block devices,
but it is easy to generalize the proof also for RBM) that replacing 
a size $k$ buffer by a size $k'$ buffer, for 
$k' = \left(1 - \frac{O(1)}{\ln k}\right)\cdot k$, increases the
optimal cost by a constant factor. 

Intuitively, resource augmentation
equips the online algorithm with some lookahead of $k-k'$ items,
and this lookahead allows the algorithm to decide on up to $k-k'$ 
steps into the future. Technically, this helps implement the online
primal-dual schema in the two crucial parts of the argument.
Firstly, the gap enables an initialization of the primal variables
(triggered by the corresponding dual constraint becoming tight)
to roughly $\frac{k-k'}{k}$ instead of $\frac 1 k$, and this reduces
the competitive ratio that one can hope for from $\log k$ to 
$\log(k/(k-k'))$.
Setting the multiplicative weights update method to prove this claim
is the main contribution of~\cite{ACER12} to extending the method
to handle RBM-style problems.
Secondly, lookahead helps bound the rate of
growth of the primal cost. In order to relate it to the rate of growth
of the dual cost, we'd like to remove items (fractionally) from the
buffer at a rate that is roughly the fractional volume that we've
already removed from items that are still present with some
weight in the buffer. This is easy if we only take into account
the removed volume that already disappeared from the buffer.
However, past decisions and current decisions extend into the
future, and there is some volume that has been scheduled to be
removed, but its removal hasn't yet happened. In particular,
some items can no longer ``participate in the game" despite
being still in the buffer, because they've already been scheduled
to be removed entirely from the buffer (this reduces the growth
rate of the dual cost). At least in the
case of block devices, if there are never more than $O(k-k')$
items of a single color in the buffer, then this scheduled but
not removed volume never exceeds $O(k-k')$, and this bounds
the growth of the primal cost adequately. However, even in the
case of block devices, the buffer may contain very large color
blocks. In~\cite{ACER12} they overcome this problem by generating
an infeasible primal solution. Thanks to the fact that in the
block devices setting the future schedule does not include any
future items, they are still able to round the infeasible solution
to get a feasible integral solution.

This approach does not work in the case of RBM. It can be made
to work if the buffer never contains a color block larger than
$O(k-k')$, as in this case a feasible primal solution can be
generated.
On the other hand, if {\em every} color block in the buffer
always has (at decision points) at least $k-k'$ items, then
(a simple version of)
the deterministic dual fitting algorithm of~\cite{AR10} (and
probably also earlier algorithms) guarantees a competitive
ratio of $O(\log(k/(k-k')))$. Thus, the main challenge and
the main technical contribution of this paper, is to combine
the two methods to give an $O(\log(k/(k-k')))$ competitive
ratio without restrictions on the instance. As at any given
time the buffer can be in a ``mixed" state with both small
and large color blocks, combining these methods is non-trivial.
A covering formulation similar to the one used in~\cite{ACER12}
cannot be used, to the best of our knowledge, so we are
required to deal with a non-covering formulation that is
harder to incorporate into the primal-dual schema. Adding 
to the challenge is the fact that the dual fitting argument
inherently generates an integral solution, whereas the 
primal-dual schema inherently
generates a fractional solution. Thus, we have to decide if
a color block will turn out to be small or large before we know
how many items of this color we can accumulate in the buffer.
If we start removing it fractionally using the primal-dual schema, 
we cannot regret this decision later and switch to dual fitting.
Rounding also poses its own challenge, mainly due to the
above-mentioned feature of RBM, whereby premature eviction
of a color may cost us a great deal later. (This is something
that does not happen in the block devices problem.)

The rest of the paper is organized as follows. After
introducing some notation, definitions, and an
overview in Section~\ref{sec: preliminaries}, we present
our online primal-dual algorithm in Section~\ref{sec: online lp} 
and analyze it in Section~\ref{sec: frac analysis}. Finally,
we present the online rounding algorithm and analyze
it in Section~\ref{sec: online rounding}. The explicit 
constants in the rest of the paper are somewhat
arbitrary. We made no attempt to optimize them.

\section{Preliminaries}\label{sec: preliminaries}

Let ${\cal I}$ be a sequence of colored items. We denote
the color of an item $i$ by $c(i)$. Abusing notation, we
denote the color of a sequence $I$ of items of the same
color by $c(I)$. We denote
by $\opt_k({\cal I})$ the cost of an optimal (offline) RBM 
schedule of ${\cal I}$ using a buffer of size $k$. The
following lemma is adapted from~\cite{ERW09,ACER12}.
For completeness, we include a proof in the appendix.
\begin{lemma}\label{lm: resource augmentation}
For every input sequence ${\cal I}$ and for every $k'<k$,
$\opt_{k'}({\cal I})\le\frac{2k+(k-k')\ln k'}{k'}\cdot \opt_k({\cal I})$.
\end{lemma}
In our algorithm and analysis we use this lemma with
$k' = k - \frac{2k}{\ln k}$, which increases the optimal
cost by a constant factor.

Consider a sequence $I$ of items of a single color $c$ in 
${\cal I}$ that includes all the items of this color between
the first and last item of $I$. If there is an RBM solution
that outputs $I$ starting at time $j$, we call the pair $(I,j)$
a {\em batch}. Thus, an RBM solution consists of scheduling
or packing batches in the interval of output time slots
$\{k+1,k+2,\dots,k+n\}$, where every output slot is used by at most
one batch, and every input item is scheduled or covered
by at least one batch. In other words, an RBM solution
is a bipartite matching of input items to output slots. The
matching must observe the order of input on each color
separately, and an item cannot be matched to an output
slot that precedes its arrival. The cost is the number of batches, where
a batch is a maximal output interval that got matched to
a set of items of the same color. This discussion leads us
to a natural linear programming relaxation for RBM. We
can think of the output slots as a channel of width $1$
spanning the output interval $\{k+1,k+2,\dots,k+n\}$.
We pack batches fractionally in this channel, without
violating the width constraint, but covering all input items 
at least once. This relaxation is essentially identical to the
one used in~\cite{AR10,AR13}. We denote it simply by
$\LP_k$. Formally, $\LP_k$ is
\begin{eqnarray}
\nonumber
\hbox{minimize\ } \displaystyle{\sum_{(I,j)} x_{I,j}} \hbox{   subject to} & &\\
\label{eq: input constraints}
\displaystyle{\sum_{(I,j):\ i\in I}x_{I,j} \ge 1} & \forall i=1,2,\ldots,n & \\
\label{eq: output constraints}
\displaystyle{\sum_{(I,j'): j'\le j <j'+|I|}x_{I,j'} \le 1} & \forall j=k+1,\ldots , k+n & \\
\nonumber
x\ge 0. & &
\end{eqnarray}
Here, $x_{I,j}$ is the weight of the batch $(I,j)$ in the packing.
Constraints~\eqref{eq: input constraints} require that every item
is eventually removed from the buffer (in batches of total weight
$1$). Constraints~\eqref{eq: output constraints} restrict the output
to remove a total weight of at most $1$ in each time slot.
The dual linear program, which we denote by $\DP_k$ is
\begin{eqnarray}
\nonumber
\hbox{maximize} \displaystyle{\sum_{i=1}^{n} y_i - \sum_{j=k+1}^{k+n}z_j} \hbox{   subject to} & & \\
\label{eq: dual constraints}
\displaystyle{\sum_{i\in I}y_{i}-\sum_{j'=j}^{j+|I|-1}z_{j'} \le 1} & \forall (I,j) & \\
\nonumber
y,z\ge 0. & &
\end{eqnarray}

Our algorithm computes online an $\LP_k$ feasible solution
$x$ and a $\DP_{k'}$ feasible solution $(y,z)$. The algorithm
feeds $x$, as it is being produced, to an online ``rounding"
procedure that produces an $\LP_k$ feasible integer solution
$\bar{x}$, which is the output of our online algorithm. Our
main result, which the rest of the paper builds towards, is
\begin{theorem}
There is an $O(\log\log k)$-competitive randomized online
algorithm for RBM.
\end{theorem}

\begin{proof}
Theorem~\ref{thm: main} establishes that the value of $x$
is at most $O(\log\log k)$ times the value of $(y,z)$, which
is a lower bound on $\opt_{k'}$, and hence at most
$O(\opt_k)$ (by Lemma~\ref{lm: resource augmentation}). 
Lemma~\ref{lm: rounding main} establishes that
the value of $\bar{x}$ is at most $O(1)$ times the value
of $x$, and this concludes the proof of the theorem.
\end{proof}

\section{The Online LP Solution}\label{sec: online lp}

In this section we give an online algorithm that constructs
a primal feasible solution $x$ to $\LP_k$ and a dual feasible
solution $(y,z)$ to $\DP_{k'}$. In Section~\ref{sec: frac analysis}
we prove the following theorem.
\begin{theorem}\label{thm: main}
$$
\sum_{(I,j)} x_{I,j}\le O(\log\log k)\cdot
\left(\sum_{i=1}^{n} y_i - \sum_{j=k'+1}^{k'+n}z_j\right).
$$
\end{theorem}

We construct simultaneously a feasible primal solution
$x$, an infeasible dual solution $(\hat{y},\hat{z})$,
and an auxiliary dual penalty $\bar{y}$. The construction
of $(\hat{y},\hat{z})$ uses a non-trivial implementation of 
the multiplicative weights update method, and $\bar{y}$ is
generated by a dual fitting argument.
A feasible dual solution $(y,z)$ can be derived by scaling
down $(\hat{y}+\bar{y},\hat{z})$ by a factor of $O(\log\log k)$.
The algorithm maintains throughout its execution for
every color $c$ an index $s_c$ which is the earliest item
of color $c$ whose primal constraint~\eqref{eq: input constraints}
is violated, i.e., $\sum_{(I,j):\ s_c\in I} x_{I,j} < 1$.

Notice that if a color $c$ is not present in the buffer (for
instance, $c$ has not been encountered yet), the algorithm
may not know $s_c$. However, if the buffer has no item of
color $c$, then the algorithm does not use $s_c$, so this
does not cause a problem. The algorithm further maintains
the  earliest output slot $t$ whose primal
constraint~\eqref{eq: output constraints} is not tight, i.e.,
$\sum_{(I,j):\ j\le t < j+|I|} x_{I,j} < 1$.
Initially, $t$ is set to $k+1$.

The dual solution $(\hat{y},\hat{z})$ is generated as
follows. Initially, all dual variables are set to $0$. The
solution is parametrized by $\mu$, which is raised at
a uniform rate. We occasionally refer to $\mu$ as
{\em time}, but this should not be confused with the
discrete input and output time steps. Further notice
that even though for convenience we describe the
algorithm as a continuous process, it can be discretized
easily, and it can be implemented 
efficiently (regardless, competitive analysis is not concerned
with computational efficiency). 
The algorithm raises all the variables
$\hat{y}_i$ with $i\ge s_{c(i)}$ and all the variables
$\hat{z}_j$ for $j\ge t$ at the same rate $d\mu$.
(This raises also future $\hat{y}_i$-s and $\hat{z}_j$-s; when 
we reach them, we will initialize their value to what's determined
by this process.)
Notice that we raise the $\hat{y}_i$-s corresponding to
violated primal constraints~\eqref{eq: input constraints},
and the $\hat{z}_j$-s corresponding to primal
constraints~\eqref{eq: output constraints} that are not tight.
Raising dual variables causes $x$ to change, thus removing items
fractionally or integrally from the algorithm's buffer. This eventually
increments the $s_c$-s and $t$, thus changing the set of dual variables that
are raised. It also affects $\bar{y}$. The process ends when $t$
passes past time $k'+n$. At this point, we simply evict the remaining
buffer contents, using the output slots up to time $k+n$.
(Notice that this last step does not cost more than the total
number of colors plus one; the total number of colors is a lower 
bound on the optimal cost.)

We now explain how raising $(\hat{y},\hat{z})$ affects
$x$. At any given time, let $B$ denote the set of
items encountered so far, whose primal constraints are
violated, and let $B_c$ denote the set of items of color
$c$ in $B$. In other words, $B_c$ includes all the color
$c$ items that appear in the input sequence starting from
$s_c$ and before the current slot $t$.
Notice that for every item in $B$, at least a
fraction of that item is still in the algorithm's buffer.
There may be additional items in the algorithm's buffer.
These are items that are already scheduled to be removed
entirely from the buffer, but $t$ hasn't yet passed
the point where they disappear from the buffer. The items
in $B$ are endowed with one of three states: {\em fractional},
{\em integral}, or {\em frozen}. If any item in $B_c$ is integral,
then they all are. Otherwise, the first ones are fractional and the 
remaining ones (if any) are frozen. We
will refer to a set of items in $B$ of the same color and the same
state as a {\em block}. Thus, $B_c$ consists of either one or two
blocks: an {\em active} block $B_c^{act}$ that is either fractional
or integral, and a frozen block $B_c^{frz}$ that might be empty
(and must be empty if $B_c^{act}$ is integral). With a slight abuse
of terminology, we sometimes also refer to all of $B_c$ as a block.
These sets ($B$, $B_c$, $B_c^{act}$, $B_c^{frz}$) are all functions 
of $\mu$ (and so are $s_c$, $t$, $x$, $\hat{y}$, $\hat{z}$, $\bar{y}$, 
and other variables defined below).

Consider a dual constraint indexed $(I,j)$ and put $c = c(I)$. Let
$$
\sigma_{I,j} = \sum_{i\in I} \hat{y}_{i}-\sum_{j'=j}^{j+|I|-1} \hat{z}_{j'}
$$
denote the current {\em dual cost} of the batch $(I,j)$. Notice that we 
know this value at any time $\mu$, even if the batch is matched to 
output slots we haven't yet reached.
\begin{fact}
Consider a batch $(I,j)$ of color $c$. If
$\frac{d\sigma_{I,j}}{d\mu} > 0$ then there must be an item
$i\in B_c\cap I$ that is matched by $(I,j)$ to an output
slot before the current time $t$.
\end{fact}

\begin{proof}
If all the items in $B_c\cap I$ are matched by $(I,j)$ at time $t$
or later, then for all $j'\in [j,j+|I|-1]$, we have that $j'\ge t$,
so $\hat{z}_{j'}$ increases, and therefore $\sigma_{I,j}$ cannot 
increase.
\end{proof}

The algorithm produces a primal solution $x$ by scheduling
batches of items, i.e., by raising $x_{J,t}$, for some batches
$(J,t)$, where $t = t(\mu)$. If we schedule a batch $(J,t)$ of 
color $c$, then $J$
begins with the items in $B_c^{act}$ (at the time $\mu$ when
$(J,t)$ is scheduled) and $J$ could extend beyond
$B_c^{act}$. We append a new item $i$ of color $c$ to $J$ if
and when the following becomes true: $i$'s state is the same as 
the state of the previous 
items in $J$ (when they were added to $J$), and we did not
pass beyond the end of the current schedule of $J$. 
In particular, when an item extends $J$ it 
is in $B_c^{act}$. Specifically, we do not append $i$
to $J$, even though it is in our buffer by the time we reach
the end of the current schedule of $J$, if $i\in B_c^{frz}$
at that time.
To summarize, scheduling a batch of color $c$ at time $\mu$
involves packing in the output stream, starting with output slot 
$t = t(\mu)$, 
the sequence of items in $B_c^{act} = B_c^{act}(\mu)$ (with a weight
that cannot be greater than the remaining unscheduled weight of the 
first item in $B_c^{act}$), and later possibly extending this sequence 
with new items on-the-fly.

The regular execution of the algorithm is to schedule
continuously batches for every color $c$ for which $B_c^{act}$
is fractional. The rate $dx_{J,t}$ at which we raise $x_{J,t}$ is
governed by pseudo-dual cost variables $\hat{\sigma}_{I,j}$
and pseudo-primal variables $\hat{x}_{I,j}$, defined for all
batches $(I,j)$. We maintain the equation
$$
\hat{x}_{I,j} = \left\{\begin{array}{ll}
                              \frac{1}{\ln k}\cdot\hat{\sigma}_{I,j} & \hat{\sigma}_{I,j} < 1, \\
                              \frac{1}{\ln k}\cdot e^{\hat{\sigma}_{I,j} - 1} & \hat{\sigma}_{I,j}\ge 1.
                              \end{array}\right.
$$
We set
$$
\frac{dx_{J,t}}{d\mu} = \max\left\{\frac{d\hat{x}_{I,j}}{d\mu}:\ c(I) = c(J)\right\}.
$$
Notice we schedule batches simultaneously for all colors with
fractional items in the buffer.

In order to complete the description of the algorithm's regular
execution, we need to explain how $\hat{\sigma}_{I,j}$ changes.
Initially, $\hat{\sigma}_{I,j}$ is set to $0$, and at certain events
(see below) we reset $\hat{\sigma}_{I,j}$ to $0$. During an interval 
$[\mu_1,\mu_2]$ with no reset, $\hat{\sigma}_{I,j}$ does not
decrease. In order to explain the increase in $\hat{\sigma}_{I,j}$,
consider the increase in $\sigma_{I,j}$ when $\mu$ changes by
an infinitesimal amount $d\mu$. Let $t = t(\mu)$, and let
$c = c(I)$. If $t \ge j$ or
if all the items in $B_c\cap I$ are matched by $(I,j)$ at time $t$
or later, then $\sigma_{I,j}$ does not increase, and $\hat{\sigma}_{I,j}$
does not change. Otherwise, $\sigma_{I,j}$ increases by $d\mu$
times the number of items in $B_c\cap I$ that are matched by
$(I,j)$ before time $t$. In this case, $\hat{\sigma}_{I,j}$ increases
by $d\mu$ times the number of items in $B_c^{act}\cap I$ that
are matched by $(I,j)$ before time $t$. We will later see that
in this case $d\hat{\sigma}_{I,j}\ge\frac{10}{11}\cdot d\sigma_{I,j}$.
We say that a batch $(J,t)$ of color $c$ that is scheduled during this 
increase is {\em relevant to} (the dual cost of the batch) $(I,j)$.
Intuitively, if $(J,t)$ is relevant to $(I,j)$, then when $x_{J,t}$ increases,
$\hat{x}_{I,j}$ increases by at most the same amount. Notice that
if a batch $(J,t)$ that is relevant to $(I,j)$ is scheduled without
interruption (i.e., it never reaches an item that is frozen at the
time slot it needs to be scheduled), then $J$ includes the last
item of $I$.

Occasionally during regular execusion, we reset $\hat{\sigma}_{I,j}$ 
to $0$. We call this a {\em regular reset} (to distinguish it from
other resets that happen when regular execusion is interrupted).
This happens in the following situation. Let the current time 
be $\mu$. Let $f = f(I,j)\in I$ be the first item that {\em interrupts}
(is not appended to) a scheduled batch $(J,t')$ ($t' < t(\mu)$)
that is relevant to $(I,j)$, because $f\in B_c^{frz}$ when
it needs to be appended. If at time $\mu$ the number of
items in $B_c$ that arrived before $f$ just dropped below
$\frac 1 2 |B_c|$, we reset $\hat{\sigma}_{I,j}$ to $0$. Notice 
that we do this only for the first such item $f\in I$, so for
any batch $(I,j)$, we do a regular reset at most once.
We denote the time of the regular reset by $\mu_0(I,j)$. 
If $(I,j)$ never experiences a regular reset, we put 
$\mu_0(I,j) = \infty$.
Also notice that if $\hat{\sigma}_{I,j}$ is reset to $0$, 
automatically $\hat{x}_{I,j}$ is reset to $0$.

Regular execution is interrupted in a few cases as follows. Upon
interruption, we keep executing the valid cases until none of them
hold, in which case regular execution is resumed.

{\em Case 1:}\/ A primal constraint~\eqref{eq: input constraints}
becomes satisfied, i.e., for some $i\in B$, $\sum_{(I,j):\ i\in I} x_{I,j}$
reaches $1$. In this case we increment $s_{c(i)}$. Notice that this also
changes $B_c^{act}$.

{\em Case 2:}\/ A primal constraint~\eqref{eq: output constraints}
becomes tight, i.e., $\sum\{x_{I,j}:\ j\le t < j + |I|\}$ reaches $1$.
In this case, we increment $t$. Each new item that enters the
buffer initializes its state as follows. If there are integral
items of the same color, it enters the buffer as integral.
Otherwise, it enters the buffer as frozen (however, the
frozen state may change immediately due to the application
of one of the following cases).

{\em Case 3:}\/ If $|B_c^{frz}| > \frac{k}{100\ln k}$ for
some color $c$, we schedule all the remaining
volume of $B_c^{act}$. (This may involve scheduling
several distinct batches, and it also increments $s_c$.) Then,
we change the state of all the items in $B_c^{frz}$ to integral.
(In particular this moves all of them to $B_c^{act}$.)
Finally, we reset $\hat{\sigma}_{I,j}$ (and hence $\hat{x}_{I,j}$) 
to $0$ for all batches $(I,j)$ of color $c$.

{\em Case 4:}\/ If $B_c^{act}$ is fractional and
$|B_c^{act}| < \frac{k}{10\ln k}$, we change the
state of all the items in $B_c^{frz}$ to fractional
(in particular, they move to $B_c^{act}$).

{\em Case 5:}\/ There is an integral block $B_c^{act}$,
and $\hat{\sigma}_{I,j}$ reaches $1$ for a color $c$
batch $(I,j)$. (Notice that taking into account the reset
in Case~3 above, we can assume that $I$ is a subset of 
$B_c^{act}$.)
We {\em suspend} all fractional scheduled batches
that haven't yet ended. We set
$\bar{y}_i = \frac{1}{2|B_c^{act}|}$ for every $i\in B_c^{act}$.
We schedule, starting at the current $t$, an integral
( weight-$1$) batch with all the items in $B_c^{act}$ followed
by any items that can be appended to that block while it
is being evicted from the buffer. We reschedule
the unfinished portion of the suspended batches following
this integral batch. Finally, we reset $\hat{\sigma}_{I,j}$ 
(and hence $\hat{x}_{I,j}$) to $0$ for all batches $(I,j)$ of 
color $c$. Notice that following this case, both
$s_c$ and $t$ are incremented by at least $|B_c^{act}|$.

{\em Case 6:}\/ Since the last application of this case, we've
moved past the end of regular execution fractionally scheduled 
batches of color $c$ with total weight at least $\frac{1}{10}$
(suspended batches are not considered to have ended).
We apply the same procedure as in Case~5 to the block
$B_c^{frz}$, except that we don't raise $\bar{y}$. To
distinguish batches scheduled by this case from integral
batches scheduled by Case~5, we will refer to the ones
we schedule here as {\em weight-$1$ fractional batches}.

\section{Analyzing the LP Algorithm}\label{sec: frac analysis}

We first observe that by the definition of the algorithm
(Case~$1$ and Case~$2$) it constructs a feasible fractional
solution.
\begin{observation}\label{obs: x feasible}
The primal solution $x$ is a feasible solution of $\LP_k$.
\end{observation}

We now bound the total volume of items in the buffer that 
the algorithm schedules at any given time while in regular
execution.

\begin{claim}\label{cl: B_c bound}
If $B_c^{act}$ is fractional, then $|B_c |< \frac{12k}{100\ln k}$
and $|B_c^{act}| < \frac{11k}{100\ln k}$.
\end{claim}

\begin{proof}
By Case~3, if $B_c^{act}$ is fractional, then 
$|B_c^{frz}|\le\frac{k}{100\ln k}$. By Case~4, we keep
new items of color $c$ in $B_c^{frz}$, unless
$|B_c^{act}| < \frac{k}{10\ln k}$. If $B_c^{act}$
drops below $\frac{k}{10\ln k}$, we move the
items in $B_c^{frz}$ to $B_c^{act}$, adding at
most $\frac{k}{100\ln k}$ new items, so
$|B_c^{act}| < \frac{11k}{100\ln k}$ always
(while fractional).
Combining this bound with the bound on
$B_c^{frz}$, we get that 
$|B_c |< \frac{12k}{100\ln k}$.
\end{proof}

\begin{corollary}\label{cor: volume beyond t-1}
At any time $\mu$, the volume of items in $B(\mu)$ 
that is scheduled beyond time $t(\mu)-1$ is less than 
$\frac{12k}{100\ln k}$.
\end{corollary}

\begin{proof}
There is a total weight of less than $1$ of scheduled batches 
that extend to time $t = t(\mu)$ and beyond (otherwise, 
Case~2 would have happened). Each of these batches is fractional,
so by Claim~\ref{cl: B_c bound} it has less than $\frac{12k}{100\ln k}$ 
items that are in $B(\mu)$. Therefore, the total volume
that is scheduled of items in $B(\mu)$ is less than
$\frac{12k}{100\ln k}$.
\end{proof}

\begin{claim}\label{cl: scheduled volume}
Consider the (partial) solution $x$ at a time of regular 
execution of  the algorithm.
The total volume that the algorithm already scheduled
of items that are currently in $B$ is bounded by
$$\sum_{(I,j)} \left(x_{I,j}\cdot |B\cap I|\right) < |B|-k'.$$
\end {claim}

\begin{proof}
Let $\mu$ denote the current time, and let $t=t(\mu)$ denote
the current output time slot. By Corollary~\ref{cor: volume beyond t-1},
the total volume of items in $B$ that is scheduled beyond time $t-1$
is less than $\frac{12k}{100\ln k}$. The total volume that is  
scheduled before time $t$ is exactly $t-1-k$. By the definition
of $B$, exactly $t-1-|B|$ items are scheduled to be removed completely 
from the buffer. Therefore, the volume that is scheduled from items in $B$
is bounded as follows:
$$
\sum_{(I,j)} \left(x_{I,j}\cdot |B\cap I|\right) < (t-1-k) - (t-1-|B|) + \frac{12k}{100\ln k} =
$$
$$
= |B|-k+\frac{12k}{100\ln k}= |B|-\left(k' +\frac{2k}{\ln k}\right)+\frac{12k}{100\ln k} 
< |B|-k'.
$$
\end{proof}

\begin{claim}\label{cl: x-hat bounded}
For every batch $(I,j)$, it holds that $\hat{x}_{I,j}\le \frac{11}{10}$
always.
\end{claim}

\begin{proof}
Let $c=c(I)$. Notice that $\hat{x}_{I,j}$ can increase beyond $1$ only
while $B_c^{act}$ is fractional.
Consider an interval $[\mu_1,\mu_2]$ 
of uninterrupted regular execution
where $\sigma_{I,j}$ increases by $\delta$ and the output time slot is $t$.
Let $(J,t)$ be a scheduled batch relevant to $(I,j)$. Notice that
$x_{J,t}$ is at least the total increase in $\hat{x}_{I,j}$ during 
the interval $[\mu_1,\mu_2]$ where $x_{J,t}$ was set. We prove 
the claim by bounding the total increase in $x_{J,t}$ for all 
relevant $(J,t)$.

First notice that the total increase for all $x_{J,t}$ that extend all 
the way to the last item in $I$ must be at most $1$. This is because after 
an increase of $1$ the last item of $I$ is no longer in $B$, and 
$\sigma_{I,j}$ cannot increase further. (A batch that is suspended
and later rescheduled is considered here as a single batch.)
The only reason that we 
do not extend $J$ to the last item of $I$ is if some item along the 
way is frozen at the time it needs to be appended to $J$.
If when 
we reach the end of $J$'s schedule, we've accumulated a cost of 
at least $\frac{1}{10}$ of interrupted schedules since the last time
a weight-$1$ fractional batch of color $c$ was scheduled, then we schedule a 
weight-$1$ fractional batch beginning with $B_c^{frz}$. Notice that 
the items of $B_c^{frz}$ are scheduled past where they
are matched by $(I,j)$, so all the remaining items of $I$ will be 
evicted from the buffer completely, and $\sigma_{I,j}$ will not 
increase further. 

So consider the first interrupted such $(J,t)$ (so the first color $c$ 
item following $J$ is $f(I,j)$). At the time $\mu$
when the interruption occurs, all the items in $B_c^{frz}(\mu)$ are 
not appended to
any scheduled batch. Let $\mu'$ denote the time when the items
that were in $B_c^{frz}(\mu)$ are schduled to be removed
entirely from the buffer (i.e., they are removed from $B$). Every
scheduled batch that removes these items must schedule them after
time slot $t(\mu)$, so unless such a batch is interrupted, it includes 
the last item of $I$.

Suppose that between time $\mu$ and time $\mu'$ no other scheduled
batch of color $c$ is interrupted. If $t(\mu')$ is past the end of the first
scheduled batch that removes $B_c^{frz}(\mu)$, then this batch is not
interrupted and thus it includes the last item of $I$. Therefore,
none of the scheduled batches that remove $B_c^{frz}(\mu)$ are 
interrupted before they include the last item of $I$. Otherwise,
at time $\mu'$ there
must be a total weight of $1$ of scheduled batches of this color,
because each item in $B_c^{frz}(\mu)$ is scheduled with total weight
$1$ in batches that begin past $t(\mu)$, and none of these batches 
are interrupted until $t(\mu')$ (by our above assumption). In
this case, at time $\mu'$ all of $B_c^{act}(\mu')$ is scheduled to
be removed completely from the buffer. Also, it must be that
$B_c^{frz}(\mu') = \emptyset$, because if $B_c^{frz}$ was non-empty
just before $\mu'$, it is moved to $B_c^{act}$ due to Case~4 of
the algorithm (as $|B_c^{act}|$ drops to $0$). Therefore, while all
these batches are still being scheduled, any new item of color $c$
is appended to all of them and is thus removed from the buffer, so
none of them are interrupted at least until the first one ends. As the
first such batch that ends must include the last item of $I$, they
all must include the last item of $I$.

If there is an interruption between time $\mu$ and time $\mu'$,
repeat this argument for the new interrupted batch $(J,t)$ and
interrupting $B_c^{frz}$. Notice that any such interruption must
have the property that $(J,t)$ must schedule the previous $B_c^{frz}$ 
past $t(\mu)$, and therefore past where it is matched by $(I,j)$.
If we accumulate $\sum_{J,t} x_{J,t}\ge\frac{1}{10}$ of
interrupted $(J,t)$ at some point, then we schedule a 
weight-$1$ fractional batch, and by the argument above
$\sigma_{I,j}$ does not increase further. Notice that in this
case the last such $x_{J,t}\le 1$ so $\hat{x}_{I,j}< \frac {11}{10}$.
Otherwise, the weight
of interrupted $(J,t)$ is less than $\frac{1}{10}$ and the weight
of uninterrupted $(J,t)$ is at most $1$, so
$\hat{x}_{I,j}< \frac{11}{10}$. (Notice that along the way we
might have reset $\hat{x}_{I,j}$, but this can only decrease its
value, and we analyzed aggregate increase $\hat{x}_{I,j}$.)
\end{proof}

\subsection{Dual feasibility}

The main technical difficulty is to show that the dual solution that the
algorithm computes is a feasible solution.
In order to prove this, we need to show that
the constraints~\eqref{eq: dual constraints} are satisfied, namely that
for every batch $(I,j)$,
$$
\sum_{i\in I}y_{i}-\sum_{j'=j}^{j+|I|-1}z_{j'} \le 1.
$$
We consider several cases in the following claims. These cases will be
combined in the pursuing proof of Lemma~\ref{lm: (y,z) feasible}.
Consider a batch $(I,j)$. The items in $I$ are partitioned by the
algorithm's execution into segments. A segment is a maximal
substring of items with the same state when removed from the
buffer. Thus, there are alternating fractional and integral segments.
An integral segment consists of a block of items that were removed
together in a single application of Case~5 of the algorithm. In
between two integral segments (or an integral segment and an
endpoint of $I$, or two endpoints of $I$) there is a fractional
segment.

We first deal with batches that do not contain an integral
segment.
\begin{claim}\label{cl: fractional batch}
For every batch $(I,j)$ for which all of $I$ is one fractional segment,
$$
\sigma_{I,j} = \sum_{i\in I}\hat{y}_{i}-\sum_{j'=j}^{j+|I|-1}\hat{z}_{j'} 
= O(\log \log k).
$$
\end {claim}

\begin{proof}
Denote $c = c(I)$. Notice that $\sigma_{I,j}$ increases only when $t > j$
and $s_c \in I$ (this is a necessary but not sufficient condition). We bound
the total increase in $\sigma_{I,j}$, ignoring possible decreases along the
way. Therefore, we may assume that $I$ is a maximal set without an
integral segment, because extending it backwards and forwards can only
make the sum of increases larger. To see this, notice that if $t > j$, then
extending $I$ backwards adds items whose $\hat{y}$ value possibly
increases, whereas its corresponding $\hat{z}$ value remains fixed.
Extending $I$ forwards adds items whose $\hat{y}$ value definitely
increases (because $s_c\in I$), and its corresponding $\hat{z}$ value
possibly also increases.

Notice that there is at most one value $\mu_0 = \mu_0(I,j)$ of $\mu$ 
where $\hat{\sigma}_{I,j}$
(and therefore $\hat{x}_{I,j}$) is reset to $0$ while $s_c\in I$, 
because $I$ does not contain
an integral segment. Recall that $f=f(I,j)$ is the first
item in $I$ that is in $B_c^{frz}$ when we need to append it to a relevant 
scheduled batch. Then $\mu_0$ is the smallest value of $\mu$ for which
$|\{i\in B_c:\ i\ge f\}|\ge \frac 1 2 |B_c|$.

Notice that whenever $\sigma_{I,j}$
increases by $\delta$, then $\hat{\sigma}_{I,j}$ increases by at least
$\frac{10}{11}\cdot\delta$. This is because the increase in $\hat{\sigma}_{I,j}$
is incurred by all the items that increase $\sigma_{I,j}$, excluding those that
are currently frozen. However, $|B_c^{act}|\ge \frac{10}{11}\cdot |B_c|$ and
if $(I,j)$ matches any item of $B_c^{frz}$ before the current time $t$, then it
matches all the items of $B_c^{act}$ before time $t$ as well. (As $B_c^{act}$
is fractional and $(I,j)$ is maximal, $B_c^{act}\subseteq I$.)

Notice that
$$
\hat{\sigma}_{I,j} = \left\{\begin{array}{ll}
                                      \hat{x}_{I,j}\cdot\ln k & \hat{x}_{I,j}\le\frac{1}{\ln k},\\
                                      1 + \ln\hat{x}_{I,j} + \ln\ln k & \hbox{otherwise.}
                                      \end{array}\right.
$$
By Claim~\ref{cl: x-hat bounded}, $\hat{x}_{I,j}\le\frac{11}{10}$ always.
Therefore, $\hat{\sigma}_{I,j} < 2 + \ln\ln k$ always. Because
$\hat{\sigma}_{I,j}$ is reset at most once, we get that 
$\sigma_{I,j}\le 2\cdot\frac{11}{10}\cdot\max\{\hat{\sigma}_{I,j}\} < \frac{22}{5} + \frac{11}{5}\cdot\ln\ln k$.
\end{proof}

The proof of the following property is useful
in the rest of the analysis.
Consider a batch $(I,j)$ of color $c$.  
Let $I_1,I_2,\ldots I_m$ be its integral segments
(by the order of the matching).
Let $j_r$ be the time that the first item of $I_r$ is matched by $(I,j)$,
and let $t_r$ be the time slot where $I_r$ was scheduled by the algorithm.
Denote by  $\Delta_r=j_r-t_r$ the difference between these times.
We also denote by $\ell_r\le|I_r|$ the number of items from $I_r$
that are in the algorithm's buffer when the algorithm decides to
remove $I_r$ (starting at time slot $t_r$).

\begin{claim}\label{cl: delta_p}
For every batch $(I,j)$ with $m$ integral segments, and
for every $1\le p<m$, we have that $\Delta_p\ge\sum_{r=p+1}^{m-1}{\ell_r}$.
\end{claim}

\begin{proof}
Notice that for every $1\le r<m$, $\Delta_r>0$. 
Otherwise the time slot where the algorithm starts removing the items 
in $I_r$ is after where they are matched by the batch $(I,j)$.
Thus, the algorithm removes all the remaining items of $I$. 
This mean that there would be no more integral segments after $I_r$. 

We now show that given $1\le r< m-1$, we have that
$\Delta_{r+1}<\Delta_{r}-\ell_r$. 
At time $t_r+|I_r|$ when we reach past the end of $I_r$'s
eviction, there are no items of this color in $B$. Therefore, all the items 
of the following fractional segment (denoted by $F_{r+1}$), 
and all the integral items that were in $B$ at the time
that segment $I_{r+1}$ was sheduled (starting at time slot $t_{r+1}$), 
enter the buffer during the input interval $[t_r+|I_r|+1, t_{r+1}]$. 
Therefore,
$t_{r+1}> t_r+|I_r|+|F_{r+1}|+\ell_{r+1}$.
Combined with the fact that,
$j_{r+1}=j_r+|I_r|+|F_{r+1}|$, we get that
$$
\Delta_{r+1}=j_{r+1}-t_{r+1}  
< j_r+|I_r|+|F_{r+1}|- \left( t_r+|I_r|+|F_{r+1}|+\ell_{r+1} \right) 
=\Delta_{r} -\ell_{r+1}
$$
This complets the proof as
$$
\Delta_p \ge \Delta_{p+1}+\ell_{p+1}\ge\Delta_{p+2}+\ell_{p+2}+\ell_{p+1}
\ge\cdots\ge\Delta_{m-1}+\sum_{p<r< m}{\ell_r}\ge \sum_{p<r< m}{\ell_r}.
$$
\end{proof}

\begin{claim}\label{cl: batch before alg}
Every batch $(I,j)$ such that the first item $i\in I$ is in $B$
at time $j$ contains a constant number of segments.
\end{claim}

\begin{proof}
Let $m$ be the number of integral segments in $(I,j)$.
 We start by showing that  $\Delta_1<\frac{11k}{100\ln k}$.
We assume that the first segment of $I$ is a fractional segment.
Otherwise, as  the first item $i\in I$ is in $B$ at time $j$,
all of $I$ is a single integral segment. 
Let $F_1$ be the first fractional segment. We also assume that 
$|F_1|> \frac{11k}{100\ln k}$ as otherwise  $\Delta_1<\frac{11k}{100\ln k}$ is triviall, 
because $t_1 >j$ and $j_1=j+|F_1|$. 
As there can be at most $\frac {11k}{100\ln k}$ fractional
items at time $j$, at least $|F_1|-\frac {11k}{100\ln k}$ items
are added to $F_1$ after time $j$. Furthermore, there are at least 
$\frac {k}{100\ln k}$ items from $I_1$ that arrive before time
slot $t_1$. These items arrive after the items of $F_1$ and therefore
after time slot $j$.
Therefore, $t_1\ge j+|F_1|-\frac {11k}{100\ln k} +
\frac {k}{100\ln k} \ge j_1 -\frac {k}{10\ln k}$.
Thus, in this case $\Delta_1 \le \frac {k}{10\ln k}$.

We now use Claim~\ref{cl: delta_p} and get that 
$$
\Delta_1\ge\sum_{1<r< m}{\ell_r}\ge (m-2)\cdot \frac{k}{100\ln k}.
$$
The upper and lower bounds on $\Delta_1$ imply that $m <13$.
\end{proof}

\begin{claim}\label{cl: batch after alg}
For every batch $(I,j)$ such that for every item $i\in I$ it holds
that $i\not\in B$ at the time it is scheduled by $(I,j)$,
$$
\sum_{i\in I}(\hat{y}_i+\bar{y}_i)-\sum_{t=j}^{j+|I|-1}\hat{z}_t =
O(\log \log k).
$$
\end{claim}

\begin{proof}
Let $(I,j)$ be a batch of color $c$ with $m$ integral segments 
that satisfies the conditions of the claim. For $i\in I$ define $M_{I,j}(i)$
to be the output time slot where $i$ is scheduled by $(I,j)$. For every item $i\in I$, let
$(\hat{y}_i+\bar{y}_i)-\hat{z}_{M_{I,j}(i)}$ be the contribution of $i$
to the pseudo-dual cost of $(I,j)$. 
Note that the sum of all contributions over $i\in I$ is exactly the left hand 
side of the claimed equation.
Let $F_p$ be the fractional segment between 
$I_{p-1}$ and $I_{p}$. We assume w.l.o.g that $I$ starts and ends with
an integral segment, as every fractional item has a negative contribution,
as we prove below.
Let $M_{alg}(i)$ be the time slot in which $i$ is removed from $B$.
Notice that for an item $i\in I_p$, $M_{alg}(i)=t_p$.
For every integral segment $I_p$, let $f_p$ be the time that the
first items of $I_p$ are ``defrosted" (i.e. they are moved from $B_c^{frz}$
to $B_c^{act}$).

We start with the following observations:\\
($i$) For every $i\in I$, $\hat{y}_i-\hat{z}_{M_{I,j}(i)}\le
-(\hat{z}_{M_{I,j}(i)}-\hat{z}_{M_{alg}(i)})$.\\
($ii$) For every $p$,
$\hat{z}_{t_p}-\hat{z}_{f_p} \ge 1/\ell_p$.\\
Observation ($i$) follows as $\hat{y}_i\le \hat{z}_{M_{alg}(i)}$.
Observation ($ii$) follows from  Case~5. When the integral block 
$I_p$ is scheduled at time slot $t_p$, the pseudo-dual cost 
$\hat{\sigma}_{I',j'}$ of some batch $(I',j')$ with $c(I') = c$ reaches $1$.
Notice that the pseudo-dual cost $\hat{\sigma }_{I',j'}$
is reset to $0$ when $I_p$ is ``defrosted" at time slot
$f_p$, and between $f_p$ and $t_p$ there are never more
than $\ell_p$ items of color $c$ in $B$. Therefore,
the rate at which $\hat{\sigma }_{I',j'}$ is raised in this
interval is at most $\ell_p\cdot d\mu$.
Therefore, $\mu$ increases by at least $1/\ell_p$ in between 
$f_p$ and $t_p$.

Let $m'$ be the maximum index for which $t_{m'}\le j$ 
(the lateset integral segment that was sheduled up to time $j$).
Let $j' = \max\{j,\max_{i\in I_{m'}} M_{alg}(i)\}$. (This is the
maximum between $j$ and the time slot where the algorithm 
removes the last item of $I_{m'}$.)  Notice that $\hat{z}_{j'} = \hat{z}_j$,
because if $j' > j$, the algorithm removes during the interval $[j,j']$
part of the integral block $I_{m'}$, and therefore the corresponding 
$\hat{z}$-s do not increase beyond their set value when the removal 
began at time slot $t_{m'}\le j$. 
We have that
\begin{eqnarray*}
& & \sum_{i\in I}(\hat{y}_i+\bar{y}_i) -\sum_{t=j}^{j+|I|-1}\hat{z}_t \\
& = & \sum_{i\in I:\ M_{alg}(i) \le j'} (\hat{y}_i+\bar{y}_i - \hat{z}_{M_{I,j}(i)}) +
\sum_{i\in I:\ M_{alg}(i) > j'} (\hat{y}_i+\bar{y}_i - \hat{z}_{M_{I,j}(i)}) \\
& = & \sum_{i\in I:\ M_{alg}(i) \le j'} (\hat{y}_i+\bar{y}_i - \hat{z}_{j'}) -
\sum_{i\in I:\ M_{alg}(i) \le j'} (\hat{z}_{M_{I,j}(i)} - \hat{z}_{j'}) + \\
& & + \sum_{i\in I:\ M_{alg}(i) > j'} \bar{y}_i - 
\sum_{i\in I:\ M_{alg}(i) > j'} (\hat{z}_{M_{I,j}(i)} - \hat{z}_{M_{alg}(i)}) \\
& \le & \sum_{p=1}^{m'}\sum_{i\in I_p} (\hat{y}_i+\bar{y}_i - \hat{z}_j) -
\sum_{i\in I:\ M_{alg}(i) \le j'} (\hat{z}_{M_{I,j}(i)} - \hat{z}_j) + \\
& & + \sum_{p=m'+1}^r \sum_{i\in I_p} \bar{y}_i - 
\sum_{i\in I:\ M_{alg}(i) > j'} (\hat{z}_{M_{I,j}(i)} - \hat{z}_{M_{alg}(i)}).
\end{eqnarray*}
The second equality follows from observation ($i$). The inequality is
explained as follows. From observation $(i)$ we have that
$\hat{y}_i - \hat{z}_{M_{I,j}(i)}\le 0$.
Therefore, if $M_{I,j}(i) \le j'$, then $\hat{y}_i - \hat{z}_{j'}\le 0$,
so we can ignore these terms for fractional segments. Furthermore,
for any $i$ in a fractional segment, $\bar{y}_i = 0$.

We upper-bound the above right-hand side as follows.\\
{\em First part:}\/ Here we bound
\begin{equation}\label{eq: first part}
\sum_{p=1}^{m'}\sum_{i\in I_p} \left(\hat{y}_i+\bar{y}_{i}-\hat{z}_j\right). 
\end{equation}
Notice that we bound the contribution of all the items before 
segment $I_{m'+1}$ assuming that they are all matched by the batch $(I,j)$
to time slot $j$. (In the next part we use the negative contribution that
each of these items $i$ accumulates between time slot $j$ and time
slot $M_{I,j}(i)$.)
Consider $i\in I_p$, $p< m'$. For every 
$p < r\le m'$, we have that $M_{alg}(i) < f_r$  and $j > t_r$.
Also notice that $t_{r-1} < f_r\le t_r$.
Therefore, using observation ($ii$) we have that
$\hat{z}_j-\hat{z}_{M_{alg}(i)} \ge \sum_{r=p+1}^{m'} \hat{z}_{t_r}-\hat{z}_{f_r}
\ge \sum_{r=p+1}^{m'}\frac{1}{\ell_r}$.
Moreover, for every such segment $I_p$
there are $\ell_p$ items $i$ each with a positive contribution
$\bar{y}_i = \frac{1}{2\ell_p}$ to~\eqref{eq: first part},
and at least $\ell_p$ items with a negative 
contribution 
$\hat{y}_i - \hat{z}_j\le \hat{z}_{M_{alg}(i)} - \hat{z}_j\le -\sum_{r=p+1}^{m'}\frac{1}{\ell_r}$
to~\eqref{eq: first part}.
The last integral segment $I_{m'}$ cobntributes at most $\frac 1 2$
to~\eqref{eq: first part}. Therefore,
\begin{equation}\label{eq: first part upper bound}
\sum_{p=1}^{m'}\sum_{i\in I_p} \left(\hat{y}_i+\bar{y}_{i}-\hat{z}_j\right)\le
\frac 1 2 +\sum_{p=1}^{m'-1}\ell_p\cdot\left(\frac{1}{2\ell_p}- \sum_{r=p+1}^{m'} \frac{1}{\ell_r}\right).
\end{equation}
Let $(1\le)\ p_1<p_2<\cdots<p_s\ (<m')$ be the indices of the segments
that have a positive contribution to the right-hand side of~\eqref{eq: first part upper bound}.
For every $u=1,2,\dots,s$, 
$\frac{1}{\ell_{p_u}} > 2\cdot \sum_{r=p_u+1}^{m'} 
\frac{1}{\ell_r}\ge 2\cdot\frac{1}{\ell_{p_{u+1}}}$.
Because for every $p$ it must be that $\frac{k}{100\ln k}\le \ell_p \le k$, 
we conclude that $s = O(\log \log k)$. Each segment $I_{p_u}$ contributes
to the right-hand side of~\eqref{eq: first part upper bound} at most $\frac 1 2$.
Therefore,
$$
\sum_{p=1}^{m'}\sum_{i\in I_p} \left(\hat{y}_i+\bar{y}_{i}-\hat{z}_j\right)
= O(\log \log k).
$$

{\em Second part:}\/ In this part we bound
\begin{equation}\label{eq: second part}
\sum_{p=m'+1}^m\sum_{i\in I_p} \bar{y}_{i} -
\sum_{i\in I:\ M_{alg}(i) \le j'} (\hat{z}_{M_{I,j}(i)} - \hat{z}_j) -
\sum_{i\in I:\ M_{alg}(i) > j'} (\hat{z}_{M_{I,j}(i)} - \hat{z}_{M_{alg}(i)}).
\end{equation}
We start by noticing that for every $p=m'+1,m'+2,\dots,m$
there are exactly $\Delta_p$ items that are matched by the
batch $(I,j)$ in the interval $[t_p,j_p-1]$ (by the fact that $p > m'$
we know that $t_p > j$). These items precede
in the input any item in $I_p$, and therefore they are
scheduled by the algorithm (completely) before time slot $f_p$.
In particular, for any such item $i$ we have $M_{alg}(i)\le f_p$.
Observation $(ii)$ implies that each such item accumulates
in the interval $[f_p,t_p]$ a negative contribution to~\eqref{eq: second part}
of at least $\frac{1}{\ell_p}$. Notice that the same items might
accumulate a negative contribution from several such intervals
for consecutive $p$-s. Further notice that if $p\ge m'+2$, then
$f_p > t_{p-1} > j$.
Using Claim~\ref{cl: delta_p} to bound $\Delta_p$, we get that
\begin{eqnarray*}
\lefteqn{\sum_{p=m'+1}^m\sum_{i\in I_p} \bar{y}_{i} -
\sum_{i\in I:\ M_{alg}(i) \le j'} (\hat{z}_{M_{I,j}(i)} - \hat{z}_j) -
\sum_{i\in I:\ M_{alg}(i) > j'} (\hat{z}_{M_{I,j}(i)} - \hat{z}_{M_{alg}(i)})}\\
&&\le 1 +  \sum_{p=m'+2}^{m-1}\left( \frac{1}{2}-\frac{\Delta_p}{\ell_p}\right)
\le 1 + \sum_{p=m'+2}^{m-1} \left(\frac{1}{2}-
\frac{1}{\ell_p}\cdot\sum_{r=p+1}^{m-1}\ell_r \right),
\end{eqnarray*}
where the extra term $1$ accounts for the contributions of $I_{m'+1}$ and $I_m$.
We therefore get that only integral segments $I_p$ such that
$\ell_p > 2\cdot \sum_{r=p+1}^{m-1}\ell_r$ can have 
a positive contibution. As $\frac{k}{100\ln k}\le \ell_p \le k$,
there are at most $O(\log \log k)$ such segments, and each
segment adds at most $\frac 1 2$ to right hand side of the above
inequallity.
\end{proof}

We are now ready to prove the main lemma.
\begin{lemma}\label{lm: (y,z) feasible}
The dual solution $(y,z)$ is a feasible solution of $\LP_{k'}$.
\end{lemma}

\begin{proof}
Consider a dual constraint indexed $(I,j)$. 
We partion the pseudo-dual cost 
$\hat{\sigma}_{I,j} + \sum_{i\in I} \bar{y}_i$
of $(I,j)$
into two parts. Let $i\in I$ be the first item for which $i\in B$ at the
time it is matched by $(I,j)$. Partition $(I,j)$ into two 
sub-batches $(I_1,j)$, $(I_2,j')$ 
such that $I_1$ contains 
all the items in $I$ smaller than $i$, $I_2$ contains 
the rest of $I$'s items, and $j'=M_{I,j}(i)$.
From Claim~\ref{cl: batch after alg} the pseudo-dual cost of 
$(I_1,j)$ is $O(\log\log k)$. Any integral segment 
of $(I_2,j')$ contributes 
exactly $\frac 1 2$ to $\sum_{i\in I_2} \bar{y}_i$. Only 
the last integral segment can have a positive contribution 
to $\hat{\sigma}_{I_2,j'}$, as any integral block with positive 
contribution to a batch $(I,j)$ evicts all the remaining 
items of $I$. Any fractional segment contributes at most 
$O(\log \log k)$ to $\hat{\sigma}_{I_2,j'}$, by 
Claim~\ref{cl: fractional batch}. Therefore, as the total
number of segments in $(I_2,j')$ is bounded by an
absolute constant
(Claim~\ref{cl: batch before alg}), the total pseudo-dual 
cost of $(I_2,j')$ is also $O(\log \log k)$.
We therefore conclude that the dual solution $(y,z)$,
derived by scaling down $(\hat{y}+\bar{y},\hat{z})$ 
by an appropriate factor of $O(\log \log k)$, is feasible.
\end{proof}

\subsection{Bounding the primal cost}

Here we bound the cost of the primal solution using the
cost of the dual solution.
\begin{lemma}\label{lm: primal cost}
At the end, 
$$
\sum_{(I,j)} x_{I,j} = O(1)\cdot
\left(\sum_{i=1}^{n} \hat{y}_i + \sum_{i=1}^{n} \bar{y}_i - 
\sum_{j=k'+1}^{k'+n} \hat{z}_j\right).
$$
\end{lemma}

\begin{proof}
We partition the primal cost of the algorithm into three parts,
according to the reason for incurring the cost.

{\em Part 1 (regular execution):}\/
Consider an increase $d\mu$ in $\mu$ during regular
execution, and let $t$ be the current output time slot.
By the definition of the algorithm, the dual variables that 
are raised at time $\mu$ 
are $\left\{\hat{y}_i\right\}_{i\in B}$, $\left\{\hat{y}_i\right\}_{ t\le i\le n}$, and
$\left\{\hat{z}_j\right\}_{t\le j \le k'+n}$. Therefore,
$$
\frac{d\left(\sum_{i=1}^{n} \hat{y}_i 
- \sum_{j=k'+1}^{k'+n} \hat{z}_j\right)}{d\mu} =
|B|+n-t+1 -(k'+n-t+1)=|B|-k'.
$$
Let $x$, $\hat{x}$ be the (partial) primal and pseudo-primal 
solutions at time $\mu$.
For every color $c$ for which $B_c^{act}\ne\emptyset$ is
fractional, let $(I_c,j_c)$ be the batch of color $c$
that maximizes $\frac{d\hat{x}_{I,j}}{d\mu}$ (i.e., at time
$\mu$ the rate of increase of $\hat{x}_{I,j}$ dictates the
rate at which color $c$ is removed from the buffer).
Notice that during regular execusion, if $x_{I,j}$ increases
then $B_{c(I)}^{act}$ must be fractional. 
Thus, by definition,
\begin{eqnarray*}
\frac{d\sum_{(I,j)} x_{I,j}}{d\mu}
&=&\sum_{\hbox{\tiny fractional } c} \frac{d\hat{x}_{I_c,j_c}}{d\mu}\\
&=&\sum_{c:\ \hat{\sigma}_{I_c,j_c}<1} 
\frac{1}{\ln k }\cdot\frac{d\hat{\sigma}_{I_c,j_c}}{d\mu}
+\sum_{c:\ \hat{\sigma}_{I_c,j_c}\ge 1} 
\frac{d\hat{\sigma}_{I_c,j_c}}{d\mu}\cdot \hat{x}_{I_c,j_c}\\
&\le&\sum_{c:\ \hat{\sigma}_{I_c,j_c}<1} \frac{|B_c^{act}\cap I_c|}{\ln k}
+\sum_{c:\ \hat{\sigma}_{I_c,j_c}\ge 1}|B_c^{act}\cap I_c|\cdot \hat{x}_{I_c,j_c},
\end{eqnarray*}
We bound the first term as follows:
$$
\sum_{c:\ \hat{\sigma}_{I_c,j_c}<1} \frac{|B_c^{act}\cap I_c|}{\ln k}\le 
\frac{|B|}{\ln k}\le |B|-k'.
$$
The last inequality follows from the fact that the volume in the buffer
of items in $B$ is at least $k-\frac{12k}{100\ln k}$ (an immediate
consequence of Corollary~\ref{cor: volume beyond t-1}).
Because the volume in the buffer of items in $B$ 
is at least $k-\frac{12k}{100\ln k}$
also $|B|\ge k-\frac{12k}{100\ln k}$. So,
\begin{eqnarray*}
|B|-k'&=& |B|-(k-\frac{12k}{100\ln k})+(k-\frac{12k}{100\ln k})-(k-\frac{2k}{\ln k})\\
&>&|B|-(k-\frac{12k}{100\ln k}) +\frac{k}{\ln k}\\
&>& \frac{|B|-(\frac{12k}{100\ln k}) +k}{\ln k}\\
&>& \frac{|B|}{\ln k}.
\end{eqnarray*}

We now bound the second term
$\sum_{c: \hat{\sigma}_{I_c,j_c}\ge 1}|B_c^{act}\cap I_c|\cdot \hat{x}_{I_c,j_c}$.
We show that for every color $c$,
\begin{equation}\label{eq: second term}
|B_c^{act}\cap I_c|\cdot \hat{x}_{I_c,j_c}\le 
O(1)\cdot\sum_{(I,j)} \left(x_{I,j}\cdot |B_c\cap I|\right).
\end{equation}
The difficulty in proving this is the following. Any increase in
$\hat{x}_{I_c,j_c}$ lower bounds the increase in $x_{I,j}$, if
the batch $(I,j)$ is relevant to $(I_c,j_c)$ at that time. In this
case, it is possible to extend the scheduled batch $(I,j)$ to
include all the items in $B_c\cap I_c$. However, the batch
might terminate before removing all those items because it
reaches an item that is in $B_c^{frz}$ at the time it needs to
be scheduled. The regular reset of $\hat{x}_{I_c,j_c}$ takes
care of this problem, as we show below.

In order to show Inequality~\eqref{eq: second term}, we consider 
three cases, according to the current value of $\mu$. The first case 
is when $\mu < \mu_0(I_c,j_c)$ (i.e., before $\hat{x}_{I_c,j_c}$
experiences a regular reset). Notice that every batch $(I,j)$ that 
is relevant to $(I_c,j_c)$ and is scheduled starting at some time
$\mu'\le\mu$, removes more than half of the items in $B_c(\mu)$,
because at the current $\mu$ less than half of $B_c(\mu)$ arrived 
after the first item $f=f(I_c,j_c)$ that causes any interruption 
in a relevant scheduled batch. In this case,
\begin{eqnarray*}
          |B_c^{act}(\mu)\cap I_c|\cdot \hat{x}_{I_c,j_c}
&\le& |B_c(\mu)|\cdot \hat{x}_{I_c,j_c} \\
&\le& |B_c(\mu)|\cdot\sum_{\hbox{\scriptsize relevent } (I,j)} x_{I,j}\\
&\le& 2\cdot\sum_{(I,j)} |B_c(\mu)\cap I| \cdot x_{I,j}.
\end{eqnarray*}

The second case is when $\mu_0(I_c,j_c)\le \mu < \mu_1(I_c,j_c)$,
where $\mu_1(I_c,j_c)$ is the time at which $f$
is scheduled to be removed completely from the buffer.
Notice that $\hat{x}_{I_c,j_c}$ is reset at $\mu_0 = \mu_0(I_c,j_c)$, 
so $\hat{x}_{I_c,j_c}\le\sum_{(I,j)} x_{I,j}$, where the sum is taken
over relevant $(I,j)$ that are scheduled after time $\mu_0$.
Consider such $(I,j)$. 
If $(I,j)$ is never interrupted (something that
might happen if $B_c^{frz}\ne\emptyset$ at the time we reach
the end of $I$), then clearly $B_c^{act}(\mu)\cap I_c\subseteq |B_c(\mu)\cap I|$.
Otherwise, let $\mu' = \mu'_{I,j}$ denote any point in the time
interval where $(I,j)$ was scheduled (the sets don't change during that
interval). Less than half the items in $B_c(\mu_0)$ arrived 
before $f$, so this remains true also for $B_c(\mu')$. 
As $\mu < \mu_1(I_c,j_c)$, we have that $f\in B_c(\mu)$.
Set $\mu'' = \mu''_{I,j}$ to be the minimum time in $[\mu',\mu]$ when
$B_c^{frz}(\mu'')\ne\emptyset$. If no such time exists, set
$\mu'' = \mu$. Clearly, $|B_c^{act}(\mu'')| > \frac{10}{11} |B_c(\mu'')|$.
Let $F = \{f,f+1,f+2,\dots,n\}$ (i.e., the input items starting
with $f$). Now, $|B_c(\mu'')\cap F|\ge \frac 1 2\cdot |B_c(\mu'')|$,
so 
$$|B_c^{act}(\mu'')\cap F|\ge\left(\frac 1 2 - \frac{1}{11}\right)\cdot |B_c(\mu'')|
> \frac 2 5\cdot |B_c(\mu'')|.
$$
Clearly, $(I,j)$ schedules all of $B_c^{act}(\mu'')$. 
Notice that 
$B_c^{act}(\mu'')\cap F\subseteq B_c(\mu)$ and
$|B_c(\mu)| < \frac{12}{10}\cdot |B_c(\mu'')|$.
Combining everything together,
\begin{eqnarray*}
            |B_c^{act}(\mu)\cap I_c|\cdot \hat{x}_{I_c,j_c} 
& \le & |B_c(\mu)|\cdot \hat{x}_{I_c,j_c} \\
& \le & \sum_{(I,j)} |B_c(\mu)|\cdot x_{I,j} \\
&  <  & \sum_{(I,j)}\frac{12}{10}\cdot |B_c(\mu''_{I,j})|\cdot x_{I,j} \\
& \le & \sum_{(I,j)} 3\cdot |B_c^{act}(\mu''_{I,j})\cap F|\cdot x_{I,j} \\ 
& \le & 3\cdot\sum_{(I,j)} |B_c(\mu)\cap I|\cdot x_{I,j}.
\end{eqnarray*}

The last case is when $\mu\ge\mu_1(I_c,j_c)$. In this
case, consider all the scheduled batches that include $f$.
Their total weight is $1$, and they've all been scheduled
before the current $\mu$. Because $f$ interrupted a
relevant batch, all these batches must be relevant. A weight
of less than $\frac{1}{10}$ of these batches is interrupted
before time $\mu$, otherwise we would have executed
Case~6, removing all the remaining items of $I_c$. This 
contradicts the definition of $(I_c,j_c)$ as the batch that 
currently, at time $\mu$, controls the rate at which color
$c$ is evicted from the buffer. Thus, a weight of at least $\frac{9}{10}$
of the scheduled batches that include $f$ schedules at time 
$\mu$ all the items in $B_c^{act}(\mu)\cap I_c$. On the other 
hand, by Claim~\ref{cl: x-hat bounded}, 
$\hat{x}_{I_c,j_c}\le\frac{11}{10}$. Therefore,
$$
|B_c^{act}(\mu)\cap I_c|\cdot \hat{x}_{I_c,j_c}\le
\frac{11}{9}\cdot\sum_{(I,j)} |B_c^{act}(\mu)\cap I|\cdot x_{I,j}
\le \frac{11}{9}\cdot\sum_{(I,j)} |B_c(\mu)\cap I|\cdot x_{I,j}.
$$
Therefore, regardless of the value of the current time $\mu$,
$$
\sum_{c: \hat{\sigma}_{I_c,j_c}\ge 1} |B_c^{act}\cap I_c|\cdot \hat{x}_{I_c,j_c} \le
3\cdot\sum_{c: \hat{\sigma}_{I_c,j_c}\ge 1}\sum_{(I,j)} |B_c\cap I| \cdot x_{I,j}
\le 3\cdot (|B| - k'),
$$
where the last inequality follows from Claim~\ref{cl: scheduled volume}.
Thus, summing the bounds on the two terms,
$$
\frac{d\sum_{(I,j)} x_{I,j}}{d\mu}\le 4\cdot (|B| - k')\le 4\cdot
\frac{d\left(\sum_{i=1}^{n} \hat{y}_i 
- \sum_{j=k'+1}^{k'+n} \hat{z}_j\right)}{d\mu},
$$
which implies trivially that the total primal cost due to 
regular execution of the algorithm is at most $4$ times 
the dual cost.

{\em Part 2 (Case 3 and Case 5 execution):}\/
Each time an integral block is evicted (Case~5),
$\sum_{i=1}^{n} \bar{y}_i$ is raised by $\frac{1}{2}$.
Preceding each such eviction there is a specific Case~3 execution, 
 when the block became integral. 
These Case~3 and Case~5 executions incur together
a primal cost of at most $3$. (Case~3 evicts a color from
the buffer at a cost of at most $1$. Case~5 schedules an
intergal block, and may suspend fractional batches of
toal weight $1$. So the cost of Case~5 is at most $2$.)
Therefore, the total primal increment due
to Case~3 and Case~5 is at most $6\cdot\sum_{i=1}^{n} \bar{y}_i$.

{\em Part 3 (Case 6 execution):}\/ Case~6 costs at most $2$ (just like
Case~5). Each time we execute Case~6 on color $c$, we've moved past
the end of regular execution fractional scheduled batches of color $c$ with
total weight at least $\frac {1} {10}$. After the end of this eviction,
$B$ does not contain any color $c$ items, therefore the next Case~6
execution is due to distinct fractional scheduled batches.
Therefore the primal increase as a result of Case~6 is at most
$20$ times the primal increase due to regular executions.
By the above analysis of regular execution, this incurs a cost of 
at most $80$ times the dual cost.
\end{proof}

\section{Online Rounding}\label{sec: online rounding}

In this section we give a randomized online algorithm that
rounds the fractional solution $x$ to an integral solution
for the reordering buffer management problem.
The rounding algorithm presented here is inspired
by our deterministic offline rounding algorithm in~\cite{AR13}.
Here we use randomness to replace the knowledge of future
input that is needed in~\cite{AR13}. 
At each
step $t$ where our rounding algorithm needs to choose a
color to evict, it uses only the input up to time $t$ and the 
fractional solution $x$ that we computed up to time $t$. 
Thus, our randomized online algorithm for reordering buffer 
management repeats two alternating steps: ($i$) Extend the 
fractional solution deterministically up to the current time. 
($ii$) Evict from the buffer using randomness some items 
chosen based on the current partial fractional solution. This 
increments the current time to the next vacant output slot.


\subsection{The rounding algorithm}

The algorithm works in phases. The first phase begins at
time $k+1$. In the beginning of a phase, the algorithm
chooses one or more color blocks to evict, based on the 
fractional
solution $x$ that was computed up to the output time
slot $t$ where the phase begins. Then, the algorithm 
evicts the chosen
blocks, and a new phase begins. Notice that
in order to execute the next phase, we need to extend
the fractional solution $x$ to the new time slot that we 
have reached, taking into account the new input items 
that have entered the buffer during the last phase.

In choosing the colors to evict in a phase, we consider
four cases. Let $\delta > 0$ be a sufficiently small constant, 
and let $t_0$ be the starting output time slot of the current 
phase. More precisely, the fractional solution computed so
far fully uses the time slots up to at least $t_0$, whereas the 
integral solution computed so far extends up to time slot 
$t_0 - 1$.

{\em Case 1:}\/ The buffer contains an item from which 
the fractional solution removed so far a weight of at least 
$\delta$. We evict the color block of this item.

{\em Case 2:}\/ The total weight of the items that the
fractional solution schedules in the time slot $t_0$ 
and are also in our buffer is at least 
$2\delta$. We choose one such item at random with 
probability proportional to the weight it is removed at
time $t_0$, and we evict its block.

{\em Case 3:}\/ A weight of more than $\frac 1 2$ of
the items that the fractional solution schedules at time
$t_0$ belong to a single color $c$ that we just evicted
from our buffer (i.e., the integral solution evicts at time 
slot $t_0 - 1$ an item of color $c$). 
In this case we first choose color
blocks to evict according to the following procedure, 
and then we evict all these blocks in arbitrary order.

We now describe the procedure for choosing color
blocks to evict in Case~3. Besides choosing color
blocks, the procedure also ``locks" some volume
fractionally scheduled before time $t_0$. Any volume
that is fractionally scheduled starts unlocked. Locked
volume is assigned to a specific evicted block, and
when the weight in the fractional buffer of an item 
in this block drops below $1 - \delta$, the volume
assigned to this block becomes unlocked again.
 
We partition the colors into classes according to the
number of items in the buffer of each color at time $t_0$.
A color $c$ is in class $s=1,2,\dots,\log k + 1$ iff the
number of items in the buffer of color $c$ is in $[2^{s-1},2^s)$.
Next, we partition the classes into subclasses as follows.
For every color $c$ let $w_c$ denote the average over the 
color $c$ items in the buffer of the unlocked volume
that the fractional solution scheduled for this item
before time $t_0$. Let $W_s$ denote the sum of $w_c$
over all colors $c$ in class $s$. To construct a subclass, 
we collect
blocks until their total $w_c$ weight exceeds $\delta$.
In a class, we construct disjoint subclasses using this
process while the remaining weight is at least $\delta$.
Notice that because $w_c < \delta$ for every color $c$,
the total weight of a subclass is in $[\delta,2\delta)$.
Also notice that in each class we might have colors with
total $w_c$ weight of less than $\delta$ that are not
assigned to subclasses. We ignore those colors. If 
$W_s < \delta$ then no block of class $s$ is chosen.
In each subclass, we choose at random one color block 
to evict. The probability of choosing a color $c$ is 
proportional to $w_c$. The chosen block locks all
the unlocked volume in the subclass. Finally, we also 
choose the largest color block in our buffer (This block 
takes care of the excess weight that we ignored in the 
above choice.)

{\em Case 4:}\/ If all else fails (i.e., for all previous
cases, the conditions for executing the case do not
hold), we choose the largest
and second largest color blocks, and also apply the
Case~3 procedure that chooses more colors. If after 
evicting the largest or second largest color block one 
of the other cases applies, 
we terminate the phase without evicting the remaining
chosen blocks. (If we don't get to evict the Case~3
procedure choices, we annul the locks generated by 
the choice.) We stress that we choose all the blocks to 
evict in this case according to the situation at time $t_0$,
but some of the chosen blocks might end up not being
evicted.

\subsection{Performance guarantees}

We show that the cost of the integral solution generated
by the rounding algorithm is within a factor of $O(1)$ of
the cost of the fractional solution generated by the
primal-dual algorithm. The main idea of the proof is the
following. Evicting a color block increases the cost of the
integral solution by $1$, and we would like to change this
cost against an increase by some (small) constant of the
cost of the fractional solution. The blocks evicted due to 
the procedure in Case~3, excluding the eviction of the largest 
block, are handled separately (see Claim~\ref{cl: case 3 procedure}).
All the remaining evictions amount to a constant number of
blocks evicted per phase. We show that for an expected
constant fraction of the phases, we can find batches that
were scheduled by the fractional solution with the following
properties: ($i$) These batches do not stretch beyond the time
slot reached by the integral solution in the corresponding phase.
($ii$) Their total weight is at least $\delta$. ($iii$) They were
not selected more than once in previous phases (excluding the 
charging of the Case~3 procedure). This, together with the
Case~3 charging scheme, implies the following guarantee.
\begin{lemma}\label{lm: rounding main}
The expected cost of the solution generated by the
rounding algorithm is $O(1)\cdot \sum_{I,j} x_{I,j}$,
where $x$ is the primal solution generated by the
primal-dual algorithm.
\end{lemma}

\begin{proof}
We consider the four cases that define a phase that
begins at time $t_0$. In the first two cases our 
charging scheme is easy to achieve. In Case~1, 
an item $i$ that is evicted in this phase is scheduled 
in the fractional solution before time $t_0$ in batches 
of total weight at least $\delta$. Because each such 
batch matches $i$ to an output slot before $t_0$, all 
of these batches end before the end of the current 
eviction of $c(i)$. So we charge this phase to the cost
of $\ge\delta$ of those batches.

In Case~2, consider the fractionally scheduled
batches with an item scheduled at time $t_0$. Let $t_1$ 
be the earliest time when the subset of these batches 
that have ended by time $t_1$ has total weight of at 
least $\delta$. Thus, the weight of the subset of these 
batches that reaches time $t_1$ is at least $1 - \delta$. 
Consider the subset of batches that schedule at 
time $t_0$ an item that is in our buffer at that time. 
This subset has total weight $w\ge 2\delta$. Therefore, 
the total weight of batches that schedule at time $t_0$ 
an item in our buffer and also reach time $t_1$ is at least 
$w - \delta\ge \delta$. The probability that the algorithm 
chooses to evict a color block of one of these batches
is at least $1 - \frac{\delta}{w}\ge\frac 1 2$.
If this event happens, the current phase ends past $t_1$, and
we charge the phase to the weight of at least $\delta$ 
of batches that use $t_0$ but end at or before $t_1$.
If our choice is unsuccessful, we don't charge the phase. 
This happens with probability at most $\frac 1 2$.

If we execute neither Case~1 nor Case~2, then for
every item in our buffer, the weight of this item that the 
fractional solution scheduled before time $t_0$ 
is less than $\delta$. Also, a weight of more than 
$1 - 2\delta$ scheduled by the fractional solution at time 
$t_0$ is of items no longer in our buffer. Notice that these items
must have appeared in the input prior to their removal, so 
we've already placed them in the buffer and evicted them in 
the past.

By the definition of $t_1$, it's still true that in the fractional 
solution the total weight of batches whose schedule contains
the interval $[t_0,t_1]$ is at least $1 - \delta$.
Let $\Delta$ denote the total volume of the content difference 
between our buffer and the fractional buffer (i.e., of the items 
in our buffer the fractional buffer lacks a total volume of $\Delta$,
and symmetrically of the contents of the fractional buffer, a
total volume of $\Delta$ belongs to items we no longer hold). 
Let $t' > t_0$ denote the earliest time where at least a weight 
of $2\delta$ of the fractionally scheduled batches that reach 
$t_1$ schedule an item that arrived at time $t_0$ or later (i.e., 
items we haven't seen yet).

Assume for now that Case~3 does not hold. If our buffer at 
time $t_0$ contains one or two colors that together have more
than $t' - t_0$ items, then the eviction of the two colors chosen 
in the first step of Case~4 makes us reach $t'$. Notice that we
reach $t'$ just by removing the items of these colors that are
already in our buffer at $t_0 - 1$. However, as we evict each
color, additional items of this color that enter the buffer might 
be appended. If we reach $t_1$,
we can charge this phase as in Case~2. If we haven't reached 
$t_1$, then Case~2 now applies for the following reason: our 
buffer at time $t_0$ contains more than $t'-t_0$ items that are 
evicted. We advance beyond $t'$ by at least the number of items
that arrived after time $t_0$ that we remove. Thus, if we haven't 
reached $t_1$, there is still a weight of at least $2\delta$ of 
fractionally scheduled batches stretching to $t_1$ with the current
item in our buffer. Therefore, the next phase will be charged with
probability at least $\frac 1 2$ (because it executes either Case~1
or Case~2). We do not charge this phase.

So let's assume that there are no such colors. Let $\gamma > 0$
be a sufficiently large constant. Suppose that
$\Delta < (t' - t_0) / \gamma$. By our assumptions, between
$t_0$ and $t'$ there is a total volume 
$> (1 - 3\delta)\cdot (t' - t_0)$ that the fractional solution
schedules of items that arrived before time $t_0$. This is
because at most $2\delta (t' - t_0)$ of the volume $t'-t_0$
belongs to items arriving past $t_0$ in fractionally scheduled
batches that reach $t_1$, and another at most $\delta (t' - t_0)$
belongs to batches that don't reach $t_1$
(regardless of when their items arrived). Of this volume, more 
than $(1 - 3\delta - 1/\gamma)\cdot (t'-t_0)$ must still be in 
our buffer at time $t_0$. Consider the fractionally scheduled
batches of total weight 
at least $1 - \delta$ whose schedule contains the interval
$[t_0,t_1]$ (which includes $t'$). At least
$\frac 3 4 - \delta$ of this weight belongs to batches that 
begin with no more than $4(t'-t_0)/\gamma$ items no longer in 
our buffer at time $t_0$. Otherwise, the total volume of items
that are no longer in our buffer but are still in the fractional
buffer is 
$> \frac 1 4\cdot 4(t'-t_0)/\gamma = (t'-t_0)/\gamma >\Delta$,
a contradiction to our assumptions.

Consider these
batches of total weight at least $\frac 3 4 - \delta$. In the
interval $[t_0+4(t'-t_0)/\gamma,t']$, they contain only items
that are either in our buffer at time $t_0$ or arrive past $t_0$.
But only less than $2\delta$ of this weight belongs to batches
that contain, up to time $t'$, any item that arrives past $t_0$
(as their schedule all reach $t_1$). So there's a weight of at least
$\frac 3 4 - 3\delta$ of these batches that in the interval
$[t_0+4(t'-t_0)/\gamma,t']$ contain only items that are in our
buffer at time $t_0$. Notice that for every color that appears
in these batches, our buffer in the beginning of the phase
contains at least $(1 - 4/\gamma)\cdot(t'-t_0)$ items of this
color. Assuming that $\gamma$ is sufficiently
large, $(1 - 4/\gamma)\cdot(t'-t_0) > (t' - t_0) / 2$. 
If there two different colors, then our buffer at 
time $t_0$ contains one or two colors that together have more
than $t' - t_0$ items, a contradiction to our assumptions.
Thus, all these batches belong to the same color $c$. The number
of items of color $c$ in our buffer is at least 
$(1 - 4/\gamma)\cdot(t'-t_0) > 4(t'-t_0)/\gamma$. 

Recall that by our assumptions so far, we execute in the
current phase Case~4.
If there is a color in our buffer with more items than $c$, 
then after evicting the largest color one of the following
two possibilities happens. If we've reached $t_1$ then we
charge this phase as in Case~2. Otherwise, more than half 
the weight that the fractional solution now removes is on
items of color $c$ that we currently have in the buffer. 
Therefore, we will next execute either Case~1 or Case~2.
We do not charge this phase,
and the next phase is charged with probability at least 
$\frac 1 2$.

If we choose to evict $c$ (because it has the maximum
number of items in the buffer) and we don't reach $t_1$, 
we end up with no items of color $c$ in the buffer, and a
weight of $>\frac 3 4 - 3\delta > \frac 1 2$ is now
being removed by the fractional solution from items of
color $c$. In particular, this means that Case~3 holds, so
in the next phase we definitely will not execute Case~4 again.
(This scenario is precisely the reason for defining Case~3.)
We do not charge this phase. If in the next phase we execute 
Case~1 or Case~2, then the next phase is charged with 
probability at least $\frac 1 2$. Otherwise, in the next
phase we execute Case~3, and as we show below, a Case~3 
phase is either charged or followed by a Case~1 or Case~2 
phase, which is charged with probability at least $\frac 1 2$.

We now analyze the remaining Case~3. Recall that Case~3 is
invoked if the fractional solution removes at time slot $t_0$
a weight of at least $\frac 1 2$ of items of a color that we've
just evicted. Define $t' = t_0 + M$, where $M$ is 
defined as follows. Consider the color $c$ batches that pass
through $t_0$. (Recall that we've just evicted color $c$.)
Each of these batches begins (at time slot $t_0$) with one
or more items that we already evicted from our buffer. Define
$M$ to be the median number of such items in a batch, where
each batch has probability proportional to its fractional weight.
Notice that at time $t'$ at least a weight of $\frac 1 4$ 
of the scheduled batches remove an item that arrived 
at time $t_0$ or later. If $\delta$
is sufficiently small (so that $\frac 1 4 - \delta\ge 2\delta$), 
a weight of at least $2\delta$ of those batches
reaches $t_1$. Thus, if we've reached $t'$ without
removing any items that arrived from $t_0$ onwards,
Case~2 applies.
Moreover, in the interval $[t_0,t']$,
at least $\frac 1 4$ of the scheduled weight is on
items that we've already evicted from our buffer
before time $t_0$. This volume is held in the fractional 
buffer at time $t_0$. So $\Delta\ge (t'-t_0)/4$.
Thus, if either Case~3 holds or the assumptions under 
which we've analyzed Case~4 do not hold, we are left 
with the following situation. There is a time $t'$ such
that $\Delta \ge (t' - t_0)/\gamma$, and if we reach
$t'$ using only items currently in our buffer, then
Case~2 holds. 
By Claim~\ref{cl: case 3 procedure} below, the Case~3 
procedure chooses colors with 
more than $t' - t_0$ items that are in our buffer at time 
$t_0$. Any item that arrives after time $t_0$ that we evict
pushes us one step further beyond $t'$. Therefore,
after evicting all the Case~3 items, either we reach $t_1$
or we can apply Case~2. All but one of the color blocks 
evicted by the Case~3 procedure are charged via 
Claim~\ref{cl: case 3 procedure}.
If we reach $t_1$, the phase is charged as in Case~2.
Otherwise, the phase is not charged, but the next phase
executes either Case~1 or Case~2 and will be charged
with probability at least $\frac 1 2$. 

Concluding the analysis, in expectation at least $\frac 1 6$ 
of the phases are charged. The worst case is when repeatedly
we have a Case~4 phase followed by a Case~3 phase followed
by a Case~2 phase which is charged with probability $\frac 1 2$. 
\end{proof}

\begin{claim}\label{cl: case 3 procedure}
For every $\delta > 0$ there exists $\gamma =\gamma(\delta)$ 
such that applying the Case~3 procedure starting at time $t_0$
chooses color blocks totalling more $t'-t_0$ items in our buffer 
at time $t_0$. Excluding one block, we can charge the eviction
of each block with probability at least $\frac 1 2$ to constant
fractional cost incurred before we complete its eviction. The
same fractional cost is never charged more than once in all
Case~3 procedure invocations.
\end{claim}

\begin{proof}
We relate the charge for chosen colors to the locking of volume that 
the fractional solution removes prior to time $t_0$ of items that we 
hold at time $t_0$. Notice that the total such volume (locked and
unlocked) is precisely $\Delta$. Notice that in every $s$-subclass, 
one execution of the Case~3 procedure locks a volume of at most
$2\delta\cdot 2^s$. While this volume is locked, all the items of
the evicted color block that locked it are kept in the fractional
buffer with weight $> 1 - \delta$. Let $\Delta_F$ denote the 
portion of $\Delta$ that is unlocked. 
We start by showing that $\Delta_F$ is close to $\Delta$. More 
specifically, we show that 
$\Delta_F \geq\frac{1-5\delta}{1-\delta}\cdot\Delta$.
To show this, we consider $\Delta$ as the volume in the fractional
buffer and not in our buffer, and $\Delta_F$ as the unlocked volume 
in our buffer and not in the fractional buffer.

Notice that each of our evictions of a color block
$B$ that is scheduled before $t_0$ contributes to $\Delta$ the total 
weight in the fractional buffer of $B$'s items at time $t_0$. The
sum of these contributions is exactly $\Delta$. If our buffer at time
$t_0$ does not contain any items with volume that was locked for 
the eviction of $B$, then $B$ contributes the same amount to $\Delta$
and $\Delta_F$. Suppose our buffer does contain items
with volume locked by the eviction of $B$. The fractional solution holds 
at time $t_0$ all such items with weight more than $1 - \delta$, 
otherwise we would have applied Case~1. Moreover, since this volume
is still locked, then all the items of $B$ must be held at time $t_0$ by 
the fractional solution with weight at least $1 - \delta$ (otherwise, the
volume assigned to $B$ would become unlocked).
So, if we 
write the contribution of $B$ to $\Delta$ as $(1 - \theta) |B|$, we get 
that $(1 - \theta) |B|\ge (1-\delta) |B|$. The volume that is locked
because of $B$ contributes at most $4\delta |B|$ to 
$\Delta-\Delta_F$, because $|B|$ is at least half the maximum size of
a block in $B$'s subclass. Therefore,
$\frac{\Delta_F}{\Delta}\ge
\min_{\theta\le\delta} \frac {1-\theta-4\delta}{1-\theta}
\geq\frac {1-5\delta}{1-\delta}$.

Going back to the main argument, let $c_1,c_2,\dots,c_m$ be
the colors in some $s$-subclass, sorted by non-decreasing 
order of the time one of their current items first drops to weight 
at most $1-\delta$ in the fractional solution. Let $c_i$ denote
the color we choose from this subclass. Notice that the contribution 
to $\Delta_F$ of this subclass is at most $2\delta\cdot 2^s$,
whereas we evict at least $2^{s-1}$ items from our
time $t_0$ buffer. The total $\Delta_F$ volume unaccounted
for is less than $\sum_{s=1}^{s_{\max}} \delta\cdot 2^s\le
2\delta 2^{s_{\max}}$,
where $s_{\max}$ is the maximum participating value of $s$ 
(this includes classes from which we did not take any color
block). To handle this portion of $\Delta_F$ that is unaccounted
for, recall that we always also evict the largest color block,
whose size is at least $2^{s_{\max}-1}$. Notice that we might
be counting this block twice, once as it might have been chosen
in an $s_{\max}$-subclass, and once as the largest block.
Summarizing the argument, we have that the number of items 
we evict from our time $t_0$ buffer is at least 
$$
\frac 1 2\cdot \frac{1}{4\delta}\cdot\Delta_F \ge 
\frac{1-5\delta}{8\delta(1 - \delta)}\cdot\Delta > 
\gamma\cdot\Delta\ge t' - t_0,
$$
for an appropriate choice of $\gamma$. (The initial $\frac 1 2$
factor is for the double-counting of the largest block.)

Finally, we deal with charging the cost of evicting the colors
we choose. Consider the colors $c_1,c_2,\dots,c_m$ in an
$s$-subclass as defined above. Let $c_j$ denote the median
color in this subclass where colors are weighted by their
contribution to $W_s$. The probability that we choose
a color with index $j$ or larger is at least $\frac 1 2$. If
this happens, we charge the fractional cost of at least
$\sum_{i=1}^j w_{c_i}\ge\frac{\delta}{2}$ that generated 
the volume of colors 
$c_1,\dots,c_j$ that we are now locking. Otherwise, we
don't charge the eviction of a block from this subclass.
Notice that
by the time the block with index $j$ or larger releases the
lock, the blocks for colors $j$ or smaller have been evicted
from our buffer (because of Case~1). Therefore, this cost is 
never charged again in a Case~3 procedure. Also notice that
in expectation half of the Case~3 evictions are charged.
\end{proof}

\bibliographystyle{abbrv}

\begin{thebibliography}{}

\end{thebibliography}


\begin{thebibliography}{10}


\bibitem{Abo08}
A.~Aboud.
\newblock Correlation clustering with penalties and approximating the
  reordering buffer management problem.
\newblock Master's thesis, Computer Science Department, The {Technion} -
  {Israel} Institute of Technology, January 2008.

\bibitem{ACER11}
A.~Adamaszek, A.~Czumaj, M.~Englert, and H.~R\"acke.
\newblock Almost tight bounds for reordering buffer management.
\newblock In {\em Proc. of the 43rd Ann. ACM Symp. on Theory of Computing},
  pages 607--616, June 2011.

\bibitem{ACER12}
A.~Adamaszek, A.~Czumaj, M.~Englert, and H.~R\"acke.
\newblock Optimal online buffer scheduling for block devices.
\newblock In {\em Proc. of the 44th Ann. ACM Symp. on Theory of Computing},
 pages 589--598, June 2012.

%

\bibitem{AAABN09}
N.~Alon, B.~Awerbuch, Y.~Azar, N.~Buchbinder, and J.~Naor.
\newblock The online set cover problem.
\newblock {\em SIAM J. Comput.}, 39(2):361--370, 2009.

\bibitem{AHK12}
S.~Arora, E.~Hazan, and S.~Kale.
\newblock The multiplicative weights update method: a meta-algorithm and applications.
\newblock {\em Theory of Computing}, 8(1):121--164, 2012.

\bibitem{AKM10}
Y.~Asahiro, K.~Kawahara, and E.~Miyano.
\newblock NP-hardness of the sorting buffer problem on the uniform metric.
\newblock Unpublished, 2010.

\bibitem{AR10}
N.~Avigdor-Elgrabli and Y.~Rabani.
\newblock An improved competitive algorithm for reordering buffer management.
\newblock In {\em Proc. of the 21st Ann. ACM-SIAM Symp. on Discrete Algorithms},
  pages 13--21, January 2010.

\bibitem{AR13}
N.~Avigdor-Elgrabli and Y.~Rabani.
\newblock A constant factor approximation algorithm for reordering buffer management.
\newblock To appear in {\em Proc. of the 24th Ann. ACM-SIAM Symp. on Discrete Algorithms}.


\bibitem{BBMN11}
N.~Bansal, N.~Buchbinder, A.~Madry, and J.~Naor.
\newblock A polylogarithmic-competitive algorithm for the $k$-server problem.
\newblock In {\em Proc. of the 52nd Ann. IEEE Symp. on Foundations of Computer Science},
pages 267--276, October 2011.

\bibitem{BBN12}
N.~Bansal, N.~Buchbinder, and J.~Naor.
\newblock A primal-dual randomized algorithm for weighted paging.
\newblock {\em J. ACM}, 59(4) (Article 19), 2012.

\bibitem{BB02}
D.~Blandford and G.~Blelloch.
\newblock Index compression through document reordering.
\newblock In {\em Data Compression Conference}, pages 342--351, 2002.

\bibitem{BN09}
N.~Buchbinder and J.~Naor. 
\newblock The design of competitive online algorithms via a primal-dual 
approach. 
\newblock {\em Found. Trends Theor. Comput. Sci.},  3(2--3):93--263, February 2009.

\bibitem{CMSS10}
H.-L. Chan, N. Megow, R. van Stee, and R. Sitters.
\newblock A note on sorting buffers offline.
\newblock {\em Theor. Comput. Sci.}, 423:11--18, 2012.


\bibitem{ERW09}
M.~Englert, H. R{\"o}glin, and M.~Westermann.
\newblock Evaluation of online strategies for reordering buffers.
\newblock {\em ACM Journal of Experimental Algorithmics}, 14 (Article 3), 2009.

\bibitem{EW05}
M.~Englert and M.~Westermann.
\newblock Reordering buffer management for non-uniform cost models.
\newblock In {\em Proc. of the 32nd Ann. Int'l Colloq. on
  Algorithms, Langauages, and Programming}, pages 627--638, 2005.


\bibitem{GSV04}
K.~Gutenschwager, S.~Spiekermann, and S.~Vos.
\newblock A sequential ordering problem in automotive paint shops.
\newblock {\em Int'l J. of Production Research, 42(9):1865--1878},
  2004.

%

\bibitem{KRSW04}
J.~Krokowski, H.~R{\"a}cke, C.~Sohler, and M.~Westermann.
\newblock Reducing state changes with a pipeline buffer.
\newblock In {\em Proc. of the 9th Int'l Workshop on Vision, Modeling and Visualization},
  page 217, 2004.


\bibitem{RSW02}
H.~R\"{a}cke, C.~Sohler, and M.~Westermann.
\newblock Online scheduling for sorting buffers.
\newblock In {\em Proc. of the 10th Ann. European Symp. on Algorithms},
  pages 820--832, 2002.


\end{thebibliography}

\newpage
\appendix
\begin{proofof}{Lemma \ref{lm: resource augmentation}}
Given an an input sequance $ֿ{\cal I}$, let $\opt_k$ be an optimal solution to the reordering buffer problem
that uses a buffer of size $k$.
We define an algorithm, $\alg_{k'}$, that uses a buffer of size $k'$ and the optimal solution $\opt_k$. 
In particular, $\alg_{k'}$ is offline. (We abuse notation and denote by 
$\alg_{k'}$ and $\opt_k$ also the cost of these respective solutions.)
Observe that we may assume that after each time $\opt_k$ finishes evicting
a color, this color will not appear in the input sequence again.
(After each eviction, we can rename all the following occurences 
with a new color $c'$ without incurring any additional cost.)
We can therefore denote by color $i$ the $i$'th color that $\opt_k$
evicts.
Consider a time $t> k'$ during the execution of $\alg_{k'}$.
Denote by $f$ the minimum color in $\alg_{k'}$'s buffer.
For any color $c$, denote by $n(c)$ the number of items
of color $c$ in $\alg_{k'}$'s buffer. For any color $c\ge c_f$, 
define the potential $\phi(c)$ of color $c$ by $\phi(c)= (c-c_f+1) n(c)$.
Finally, we define for each color $c$ a counter $p(c)$ initialized
to $0$. Intuitively, this counter counts the number of items 
larger than $c$ that were evicted so far.
Notice that $c_f$, $n(c)$, $\phi(c)$, and $p(c)$ are all a functions of 
$t$. 
The algorithm $\alg_{k'}$ works as follows.\\
For any time $t$.
\begin{enumerate}
\item If the eviction of color $c_f$ will evict 
the last item of this color in $\cal{I}$, evict color $c_f$. 
\item Otherwise, let $c$ be the color with the maximum potential
in $\alg_{k'}$'s buffer. Evict exactly the $n(c)$ items of this color
currently in the buffer (without appending any aditional 
arriving items of the same color).\\
If after the eviction we cannot execute Step~1, we update $p(i)$
 for every $c_f\le i <c$, by increasing $p(i)$ by $n(c)$.
\end{enumerate}

We start by proving a bound on $p(i)$.
\begin{claim}
For any color $i$, at any time during the execution of the algorithm, 
$p(i)<k-k'$.
\end{claim}
\begin{proof}
Notice that it is sufficient to bound $p(c_f)$ in any point in time,
as this is the maximum counter among the colors that their 
$p(i)$ can still increase.
Assume for contradiction that at a given time $t$ the counter 
$p(c_f)$ became larger than $k-k'$ (right after removing a color by Step~2). 
Consider this time $t$. Let $n_1$ be the number of items $\alg_{k'}$ evicted from 
items of a color smaller than $c_f$ (this equals to the number of items with
 a color smaller than $c_f$). Let $n_2$ be the 
number of items $\alg_{k'}$ evicted from color $c_f$.
At time $k+1+n_1$, $\opt_k$ started evicting color $c_f$,
therefore at time $k+1+n_1+n_2$, if the buffer evicted at
most $n_2$ items from color $c_f$, evicting this color will reach 
the last item of $c_f$.
On the other hand, 
because $p(c_f)$ is at least the number of items from colors larger than
$c_f$ that were evicted so far,
it holds that 
$$
t \ge k'+1+n_1+n_2+p(c_f)> k+1+n_1+n_2.
$$
This is in contradiction to Step~2, as the counter is increased only if 
we cannot apply Step~1.
\end{proof}

Next we show a lower bound on the potential.
\begin{claim}
Consider a time $t$ right before executing Step~2.
The maximal potential is $max_c \phi (c)\ge \frac{k'}{1+\ln k'}$.
\end{claim}
\begin{proof}
Denote $s=\frac{k'}{1+\ln k'}$.
Assume for contradiction that for any color $c$ we have that 
$\phi (c)<s$. Therefore, $c-c_f<s$, and $n(c) < \frac{s}{c-c_f +1}$, 
for every color $c$ in the buffer.
Because there are exactly $k'$ items in the buffer, 
$$
k'=\sum_{c=c_f}^{\lfloor s\rfloor} n(c) < \sum_{c=c_f}^{\lfloor s\rfloor} \frac{s}{c-c_f +1} =
\sum_{i=1}^{\lfloor s\rfloor-c_f+1} \frac {s}{i} = s\cdot H_{\lfloor s\rfloor}\le k'.
$$
Thus, the claim follows.
\end{proof}

We are now ready to prove our lemma.
Notice that the number of times $\alg_{k'}$ executes Step~1
is at most $\opt_k$. Furthermore, notice that 
in every execution of Step~2 , except for $\opt_k$ executions, $\sum_c p(c)$  is increased 
by at least $\frac{k'}{1+\ln k'}$.
Because $\sum_c p(c)<(k-k') \opt_k$, there could be at most 
$\frac{(k-k')(1+\ln k')}{k'} \opt_k$ executions.
Therefore, 
$$
\alg_{k'} \le \left(2+\frac{(k-k')(1+\ln k')}{k'} \right)\cdot\opt_k,
$$
and the lemma then follows.
\end{proofof}

\end{document}